\documentclass[10pt]{elsarticle}
\usepackage{amssymb,amsmath}
\usepackage{array,multirow,graphicx}
\usepackage{enumitem}
\usepackage{bbm}
\usepackage{dirtytalk}
\usepackage{mathtools}
\usepackage[margin=1in,papersize={8.5in,11in}]{geometry}
\usepackage[utf8]{inputenc}
\usepackage[english]{babel}
\usepackage{amsthm}
\usepackage{float}
\usepackage{amsfonts,mathrsfs}
\usepackage{bm}
\usepackage{mathrsfs}
\usepackage{makecell}
\usepackage{tikz}
\usetikzlibrary{matrix}
\newcommand{\STAB}[1]{\begin{tabular}{@{}c@{}}#1\end{tabular}}
\newtheorem{assumption}{Assumption}

\usepackage{epstopdf}
\usepackage{subcaption}
\usepackage{tablefootnote}
\graphicspath{ {images/} }
\usepackage{comment}

\newtheorem{theorem}{Theorem}

\newtheorem{prop}{Proposition}
\theoremstyle{definition}

\def\Z{\mathbb{Z}}
\def\R{\mathbb{R}}
\def\H{\mathcal{H}}

\def\T{\mathbb{T}}

\def\Tr{\overrightarrow{\mathcal{T}}}
\def\W{\mathcal{W}}
\def\A{\mathcal{A}}
\def\B{\mathcal{B}}
\def\D{\mathcal{D}} 
\def\S{\mathcal{S}} 
\def\L{\mathcal{L}} 
\def\J{\mathcal{J}} 
\def\M{\mathcal{M}}
\def\N{\mathcal{N}}

\def\la{\langle} 
\def\ra{\rangle_{\rho}}

\def\P{\mathcal{P}}
 
\def\X{\mathcal{X}}

\def\Q{\mathcal{Q}}
\def\E{\mathcal{E}}

\def\G{\mathcal{G}}
\def\K{\mathcal{K}} 
\def\I{\mathcal{I}}

\def\Pi{\mathbf{P}}

\journal{arXiv}

\begin{document}
\begin{frontmatter}
\title{Nonequilibrium statistical mechanics for stationary turbulent dispersion}
\author[ucm]{Yuanran Zhu\corref{correspondingAuthor}}
\cortext[correspondingAuthor]{Corresponding author}
\ead{yzhu56@ucmerced.edu}
\address[ucm]{Department of Applied Mathematics, University of California, Merced\\ Merced (CA) 95343}
 
\begin{abstract}
We propose a unified framework to study the turbulent transport problem from the perspective of nonequilibrium statistical mechanics. By combining Krarichnan's turbulence thermalization assumption and Ruelle's recent work on nonequilibrium statistical mechanics settings for fluids, we show that the equation for viscous fluid can be viewed as the non-canonical Hamiltonian system perturbed by different thermostats. This allows an analogy between the viscous fluid and the  nonequilibrium heat conduction model where the Fourier modes can be regarded as the ``particles''. With this framework, we reformulate the dispersion of Lagrangian particles in turbulence as a nonequilibrium transport problem. We also derive the first and the second generalized fluctuation-dissipation relations for the Lagrangian particle using respectively the path-integral technique and the Mori-Zwanzig equation. The obtained theoretical results can be used predict the dispersion of the Lagrangian particle in a general nonequilibrium. 
\end{abstract}
\end{frontmatter}

\section{Introduction}
As Feynman remarks: \say{Turbulence is the most important unsolved problem of classical physics.} The importance stems from the fact that turbulent phenomena exists universally in our daily life wherever there is fluid with low viscosity. Mathematically, macroscopic-scale turbulence phenomena can be well-described by the Navier-Stokes (NS) equation with high Reynolds number. With the development of computational mathematics and physics, modern computing machinery enables us to numerically solve the Naiver-Stokes equation and simulate a turbulent field with very high accuracy. Although it is still a challenging task to simulate turbulence in circumstances such as the flow in a tube with complex geometry, it is fair to say that comparing to problems such as what turbulence is and how to describe it, our theoretical understanding of its physical mechanism is still quite lacking.  

Having observed the complexity of the turbulent dynamics, since 1930s many pioneers of turbulent study such as G.I. Taylor and his followers started to introduce statistical tools such as the stochastic processes and correlation functions to analyze the dynamical behaviour of turbulence. Up to now, this is still a very active research direction in fluid mechanics which can be termed as the {\em statistical fluid mechanics} \cite{monin2013statistical,pope2001turbulent}. Roughly speaking, the statistical fluid mechanics views turbulence as a random field where the randomness comes from the chaotic nature of the Navier-Stokes equation and the universally existed perturbations in the initial conditions, boundary conditions and body forces.  
The methodology we are going to use, although shares many similarities with the statistical method, takes a slightly different path and may be called as the {\em statistical mechanics/physics} method for the study of turbulence. This approach also has a long history and was developed mainly by Kraichnan \cite{kraichnan1975statistical,kraichnan1970diffusion,kraichnan1977eulerian,kraichnan1980two} and Zakharov \cite{zakharov2012kolmogorov}. Different from the first viewpoint, the statistical mechanics approach starts with the thermalization assumption for the inviscid fluid. As a Hamiltonian system with high-dimensionality, the thermalization assumption states that the inviscid fluid will undergo a dynamical process and finally reach a thermal equilibrium which be characterized by a certain statistical ensemble. Combining with the ergodicity hypothesis, the ensemble average for any physical quantity $f$ can represent the value we obtained through experimental measurement. 

The above groundbreaking work sets the theoretical basis for the study of turbulence using statistical mechanics. The early success of this framework is Kraichnan's prediction of the inverse energy cascade for turbulence, which was proved later via numerical simulations \cite{kraichnan1980two}. Further explorations in this direction have revealed many unique features of the statistical mechanics for fluid systems. First, the evolution equation for the inviscid fluid normally is a non-canonical Hamiltonian system, i.e. there is no generalized coordinate $\{p,q\}$ and the phase space dimension can be odd. In fact, relevant mathematical studies \cite{arnold1999topological,morrison1998hamiltonian} have shown that fluid systems have the Poisson structure, which is more general than the symplectic structure of the canonical Hamiltonian systems commonly studied in condensed matters. The Poission structure is normally degenerate and leads to {\em additional} symmetries and invariants for fluid systems. The existence of these additional invariants, known as the Casimir invariants, changes the statistical mechanical setting for fluids in a fundamental way. Particularly, most inviscid fluid systems have Casimir invariants \cite{morrison1998hamiltonian} which have to be included in the postulated statistical equilibrium measure, otherwise unphysical results will emerge \cite{kraichnan1980two}. Another development is study of the viscous fluid via nonequilibrium statistical mechanics. Due to the existence of the dissipative force, the viscous fluid is a typical {\em nonequilibrium} system. The work of Gallavotti et al. \cite{gallavotti1996chaotic,gallavotti1995dynamical} and Ruelle \cite{ruelle1999smooth,ruelle1978measure} discussed systematically the possibility to define a proper probability measure, now known as the Sinai-Ruelle-Bowen (SRB) measure, for such a nonequilibrium system. The connection between the fluctuation theorem with the entropy production rate was also discussed by Gallavotti et al. \cite{gallavotti1996chaotic,gallavotti1995dynamical}. In particular, they started to notice the similarities between the viscous fluid equations and some nonequilibrium models used in the statistical mechanics for molecular systems. The recent papers of Ruelle \cite{ruelle2012hydrodynamic,ruelle2014non} further show that the turbulent flow can be regarded as a heat flow problem along a chain of mechanical systems describing units of fluid of smaller and smaller spatial extent. These ideas are the direct inspiration of our work.

In this paper, we follow the line of Kraichnan-Gallavotti-Ruelle and further discuss the nonequilibrium statistical mechanics for fluids. Different from the interest of Gallavotti and Ruelle, which is the deterministic NS equation, our main concentration would be the {\em stochastic Navier-Stokes} (SNS) equation, which is the typical fluid equation driven by Gaussian white noise. The SNS equation has its own usage in the modeling of turbulence. Its mathematical difference with respect to the regular NS equation with the deterministic dissipative force $\frac{1}{Re}\Delta u$ is important for our research. Comparing to the later, the analysis for the SNS equation is relatively simpler because one can bypass the complicated discussions of the chaotic hypothesis and the Axiom A, Anosov systems which are required for the construction of the SRB measure \cite{ruelle1999smooth,gallavotti1996chaotic,young2002srb}, and use the stochastic analysis techniques to rigorously {\em prove} the existence, uniqueness and smoothness of the equilibrium and nonequilibrium steady state (NESS) measure \cite{weinan2001gibbsian,weinan2001ergodicity}. Moreover, the ergodicity of the system can be obtained using the Lyapunov function methods \cite{weinan2001ergodicity,bedrossian2019almost,mattingly2002ergodicity,hairer2006ergodicity,hairer2008spectral} for the stochastic fluid system. All these benefits enable us to get several {\em mathematically rigorous and numerically computable} results for the turbulent flow. Our main findings are as follows: (I) In the Fourier mode space, the viscous fluid equation can be regarded as a heat conduction model where the Fourier modes $u_k$ can be viewed as ``particles'' and the averaged kinetic energy $\langle u_k^2\rangle$ can be interpreted as the ``temperature''. (II) The Lagrangian particle dispersion in turbulence can be reformulated as a problem of nonequilibrium transport, which is similar to the self-diffusion of particles in a nonequilibrium molecular system. (III) The above analogies allow a heuristic derivation of the generalized fluctuation-dissipation relations (FDRs) for observables in the fluid system, in particular, for the Lagrangian particle. To this end, we use the result which are recently developed for the nonequilibrium molecular systems  \cite{baiesi2009fluctuations,baiesi2009nonequilibrium,baiesi2010nonequilibrium,caprini2021generalized,marconi2008fluctuation,zhu2021hypoellipticity3,agarwal1972fluctuation,majda2002statistical,majda2005information,zhu2021hypoellipticity3}. The heuristic derivation is later justified with rigorous mathematical proofs. In the end, we obtain a nonequilibrium version of Kubo's fomula in the sense of path integral, which can be used to predict the dispersion of the Lagrangian particle. Moreover, a generalized second FDR is established using the Mori-Zwanzig formalism, which leads to an effective model for the Lagrangian particle in terms of the generalized Langevin equation (GLE).  

The whole paper is organized as follows. Section 2 provides a short preview of the aforementioned theoretical results (I)-(III). In Section \ref{sec:Noneq_turbulence}, we first review Kraichnan's basic setting for the equilibrium statistical mechanics of turbulence and its generalization for arbitrary inviscid fluids. Secondly, we introduce our interpretation of the viscous fluid as the nonequilibrium heat conduction model in the Fourier mode space. This idea is heavily motivated by Ruelle \cite{ruelle2012hydrodynamic,ruelle2014non} and we will explain the different emphasis of our framework. In Section \ref{sec:Turb_disper_setting}, we reformulate the turbulent dispersion problem in the framework of nonequilibrium statistical mechanics and discuss its major difference with the transport in molecular systems. In Section \ref{sec:FDT}, we use the path-integral technique and the Mori-Zwanzig equation to derive respectively the first and the second generalized FDR to the turbulent transport problem. The main findings of this paper are summarized in Section \ref{sec:summary}.

\section{Brief summary of the main result}
\label{sec:sum_brief}
In this section, we briefly summarize the theoretical results obtained in this paper. Consider a stochastic Navier-Stokes (SNS) equation for the incompressible viscous fluid in the Fourier mode space:
\begin{align}\label{intro_SNS}
        du_k=B(u,u)dt-\nu|k|^2u_kdt+\sigma_kd\W_k(t),
\end{align}
where $u_k$ is the $k$-th Fourier mode of the Eulerian velocity field $u=u(x,t)$, $B(u,u)$ is a bilinear term which comes from the advection $u\cdot\nabla u$, $\nu$ is the dynamical viscosity and $\sigma_kd\W_k(t)$ is the added stochastic force. Note that the pressure term $-\nabla p$ is eliminated from the equation by using the incompressible condition $\nabla\cdot u=0$. Our aim is to use the nonequilibrium statistical mechanics to study the transport phenomena associated with the SNS equation \eqref{intro_SNS}, such as the Lagrangian particle dispersion.   

\paragraph{\textbf{Heat conduction interpretation of turbulence}} In Ruelle's recent papers \cite{ruelle2012hydrodynamic,ruelle2014non}, the deterministic Navier-Stokes (NS) equation was interpreted as the heat conduction model where the heat flows along a chain of mechanical systems describing units of fluid of smaller and smaller spatial grids. The physical picture is the same for SNS equation \eqref{intro_SNS} in the Fourier mode space. As a comparison, a commonly used heat conduction model used in statistical physics is given by the following stochastic differential equation (SDE):
\begin{align}\label{intro_compare}
\begin{dcases}
\begin{aligned}
dq_i&=p_idt\\
dp_i&=-\partial_{q_i}H dt-\gamma_ip_idt+\sqrt{2k_BT_i\gamma_i}d\W_i(t),
\end{aligned}
\end{dcases}
\end{align}
where a mechanical Hamiltonian system $H$ is in contact with Langevin thermostats with temperatures $T_i$, which are modeled by stochastic force $\sqrt{2k_BT_i\gamma_i}d\W_i(t)$. If $T_i$ does not equal to each other, the system will be in the nonequilibrium and there exists heat current flowing from the high-temperature thermostats to the low-temperature thermostats. By observing intuitively the similarity between the equation of motions \eqref{intro_SNS}-\eqref{intro_compare}, we note that Ruelle's idea can be re-formulated in the frequency domain, where the Fourier modes $u_k$ can be understood as the ``particles'' and the stochastic forces which drive the system from the equilibrium can be regarded as the thermostats. In the dynamical process of the fluid, the kinetic energies are transferred between Fourier modes $u_k$ hence defines the ``heat'' flows.  Similar analogies can be made between the heat conduction model with the isokinetic thermostats and the deterministic NS equation without forcing, and the heat conduction model with the Nos\'e-Hoover thermostats and the NS equation driven by deterministic forces. 

\paragraph{\textbf{Turbulent dispersion as a transport problem}}The above observations hint us to deal with the transport problems of turbulence using the standard tools of nonequilibrium statistical mechanics. For the heat conduction model \eqref{intro_compare}, we normally choose the averaged heat flux $J(p(t),q(t))$ as the observable to quantify the intensity of the heat transport. Specifically, $J(p(t),q(t))$ can be written as a function of the microscopic coordinate $\{p(t),q(t)\}$ as 
\begin{align*}
    J(p(t),q(t))=\frac{1}{N}\sum_jp_jV'(q_{j+1}-q_j),
\end{align*}
where $V(q_{j+1}-q_j)$ is the interactive potential energy between neighbourhood particles. The scenario is slightly different for the turbulent transport. Since Eqn \eqref{intro_SNS} is a model of the velocity field from the {\em Eulerian} description, it is difficult to write the {\em Lagrangian} observables explicitly using the phase space coordinate $\{u_k\}$. To study the dispersion of the Lagrangian particle in the turbulence, we are lead to study an {\em extended} SNS system:
\begin{equation}\label{intro_SNS_Extend}
\begin{dcases}
\begin{aligned}
        du_k&=B(u,u)dt-\nu|k|^2u_kdt+\sigma_kd\W_k(t)\\
        dX_i&=u_i(X(t),t)+\kappa_i d\D_i(t),
\end{aligned}
\end{dcases}
\end{equation}
where $X_i(t)$ is the position of a characteristic Lagrangian particle in $i$-th axis, $u_i(X(t),t)$ is the Eulerian velocity at space-time location $(X(t),t)$ and $\kappa_i$ is the molecular diffusivity of the Lagrangian particle.

\paragraph{\textbf{Generalized FDRs}}Different from the heat conduction model where the {\em equilibrium} time auto-correlation function of the heat flux can be used to calculate the conductivity in the {\em near-equilibrium} regime \cite{lepri2003thermal}, real world turbulence is almost always {\em far from the equilibrium} because of the existence of the dissipative force, which leads to a general failure of the Green-Kubo formula. Hence, inevitably one is lead to develop/use the {\em nonequilibrium} version FDRs to study the turbulent transport. To this end, we noticed the recently developed path-integral-form first FDRs for nonequilibrium systems \cite{baiesi2009fluctuations,baiesi2009nonequilibrium,baiesi2010nonequilibrium,maes2020response}. To sum up the theoretical result obtained therein, we found that the path-integral-first FDR actually yields a nonlinear response formula for observable $O(t)$ in the nonequilibrium
\footnote{We note the interesting similarity between Eqn \eqref{intro_entropy_prod_FDR} and Crooks fluctuation theorem \cite{crooks1999entropy} of the form:
\begin{align}\label{intro_crooks_FT}
\frac{P[x(t)]}{P[\tilde{x}(t)]}=e^{\sigma[x(t)] t},
\end{align}
where $P[x(t)]$ and $P[\tilde{x}(t)]$ are respectively the probability of the time-forward trajectory and its reversal. $\sigma[x(t)]$ is the entropy production rate. In fact, from the derivation of \eqref{intro_entropy_prod_FDR} we are going to present in Section \ref{sec:1st_FDR_path}, a functional Radon-Nikodym derivative (Formula \eqref{Radon-Nikodym_derivative}) similar to the form of \eqref{intro_crooks_FT} is the key relation which leads to \eqref{intro_entropy_prod_FDR}. From this point of view, the relationship between Crooks fluctuation theorem \eqref{intro_crooks_FT} and the FDR \eqref{intro_entropy_prod_FDR} can be interpreted as a manifestation of Onsager's regression hypothesis in the nonequilibrium. 
}:
\begin{align}\label{intro_entropy_prod_FDR}
    \langle O(t)\rangle_{\rho_{\delta}}=\left\langle  e^{\frac{1}{k_B}\int_0^t\delta\E(s)ds}O(t)\right\rangle_{\rho},
\end{align}
where $k_B$ is the Boltzmann constant, $\delta\E(s)$ is the change (brought by the perturbation) of the entropy production rate for a system transient from nonequilibrium state $\rho$ to another nonequilibrium state $\rho_{\delta}$. Linear response version of \eqref{intro_entropy_prod_FDR} were derived for the heat conduction model \eqref{intro_compare} in \cite{maes2020response}. Naturally, the aforementioned analogy between the heat conduction model and the SNS equation \eqref{intro_SNS},\eqref{intro_SNS_Extend} allows a heuristically derivation of a path-integral-form first FDR for the fluid observable $O(t)$. For instance, if we take $O(t)=X(t)$ as the position of the Lagrangian particle and consider the fluid viscosity perturbation $\nu'=\nu+\delta\mu$ for a 2D SNS equation, i.e. 
\begin{align}\label{intro_FDR_perturbation}
    \begin{dcases}
    d\alpha_l&=F_{l}(\{\alpha_l\},\{\beta_l\})dt-\nu|l|^2\alpha_ldt+\sigma_ld\B_l\\
    d\beta_l&=G_{l}(\{\alpha_l\},\{\beta_l\})dt-\nu|l|^2\beta_ldt+\gamma_ld\W_l
    \end{dcases}
    \quad
\xRightarrow[]{+\delta\mu\Delta \omega}
\quad
    \begin{dcases}
    d\alpha_l&=F_{l}(\{\alpha_l\},\{\beta_l\})dt-\nu'|l|^2\alpha_ldt+\sigma_ld\B_l\\
    d\beta_l&=G_{l}(\{\alpha_l\},\{\beta_l\})dt-\nu'|l|^2\beta_ldt+\gamma_ld\W_l,
    \end{dcases}
\end{align}
where Eqns in \eqref{intro_FDR_perturbation} are an equivalent expression of the turbulent flow using the vorticity field:
\begin{align*}
\omega(x,t)=\sum_{k}\alpha_k\cos(k\cdot x)+\beta_k\sin(k\cdot x),
\end{align*}
and $F_{l},G_l$ are bilinear functions of $\alpha_k$ and $\beta_k$. Then by using the extended dynamics which is similar to \eqref{intro_SNS_Extend}, we obtain the following nonlinear response formula\footnote{For turbulence, $\delta\E(s)$ has no physically meaningful unit hence we do have not normalization constant $1/k_B$ here.}:
\begin{align}\label{intro_entropy_prod_FDR_fluid}
    \langle X_i(t)\rangle_{\rho_{\delta}}=\left\langle  \N(t)X_i(t)\right\rangle_{\rho}=\left\langle  e^{\int_0^t\delta\E(s)ds}X_i(t)\right\rangle_{\rho},
\end{align}
where $\langle\cdot\rangle_{\rho_{\delta}}$ and $\langle\cdot\rangle_{\rho}$ are, respectively, the path-space ensemble average with respect to the steady state distribution of the perturbed and non-perturbed SNS equation in \eqref{intro_FDR_perturbation}. In the linear response regime, we get a nonequilibrium extension of Kubo's formula:
\begin{align}\label{intro_heat_linear_response}
\langle X_i(t)\rangle_{\rho_{\delta}}-\langle X_i(t)\rangle_{\rho}&=\int_0^t \langle \mathcal{R}(s)X_i(t)\rangle_{\rho}ds.
\end{align}
The response function $\mathcal{R}(s)$ in \eqref{intro_heat_linear_response} is explicitly given by:
\begin{align}
\mathcal{R}(s)=\sum_{l}\bigg(\frac{\delta\mu\nu|l|^4}{\sigma_l^2}-\frac{\delta\mu|l|^2}{2\langle\alpha_l^2\rangle_{\rho}}\bigg)\alpha_l^2(s)
+
\left(\frac{\delta\mu\nu|l|^4}{\gamma_l^2}-\frac{\delta\mu|l|^2}{2\langle\beta_l^2\rangle_{\rho}}\right)\beta_l^2(s)
-\frac{\delta\mu|l|^2}{\sigma_l}\alpha_l(s)B_l(s)
-\frac{\delta\mu|l|^2}{\gamma_l}\beta_l(s)W_l(s),
\end{align}
where $d\B_l(t)=B_l(t)dt$ and $d\W_l(t)=W_l(t)dt$. Response formulas \eqref{intro_entropy_prod_FDR_fluid}-\eqref{intro_heat_linear_response} are proved to hold for the equilibrium SNS equation and the nonequilibrium SNS equation driven by the {\em non-degenerate} white noise. The key technique we use is the exact probability density functional for SDEs \eqref{intro_FDR_perturbation}  and a functional Radon-Nikodym derivative, which are obtained by Stratonovich in \cite{moss1989noise} via path integrals. The first FDR is a nonequilibrium response result which can be applied to study turbulent dispersion via numerical simulations of the unperturbed dynamics. Namely, by averaging the sample paths evolving from the NESS state of the unperturbed SNS equation, one can evaluate the response $\langle X_i(t)\rangle_{\rho_{\delta}}$ as:
\begin{align}\label{intro_ensemble_average}
\langle X_i(t)\rangle_{\rho_{\delta}}\approx\frac{1}{K}\sum_{k=1}^K X_{i,k}(t)+\frac{1}{K}\sum_{k=1}^K\int_0^t \mathcal{R}_{,k}(s)X_{i,k}(t)ds,
\end{align}
where $\mathcal{R}_{,k}(s),X_{i,k}(t)$ are the $k$-th sample of the response function $\mathcal{R}(s)$ and the observable $X_i(t)$ respectively. In addition, we can derive a tentative linear response formula for the {\em deterministic} turbulence using a singular perturbation argument. A comparison of the current framework with the {\em renormalized perturbation theory} mainly developed by Kraichnan is presented in Section \ref{sec:compare}. We particularly emphasize the vital differences in terms of the research targets  and the adopted methodology.

The path-integral-form first FDR can be characterized as a result which uses the information of a NESS to obtain the information of the perturbed NESS (see Figure \ref{fig:perturbation}). From the formal expression \eqref{intro_heat_linear_response}, we see that it gives the moment information such as the mean, variance (choosing the observable to be $X_i^2(t)$) of the Lagrangian particle $X(t)$. In applications, it is often constructive to directly approximate/simulate the sample paths of $X(t)$. This is known as the reduced-order modeling problem for turbulent dispersion \cite{pope2001turbulent}. The Mori-Zwanzig(MZ) formalism provides a systematical way to build such a model from the underlying (S)NS equation. Using the MZ framework here, we are lead to an {\em exact} reduced-order model for $X(t)$ in terms of the the generalized Langevin equation (GLE):
\begin{align}\label{intro_gle}
    \frac{d}{dt}X (t)=\Omega X(t)+\int_0^tM(t-s)X(s)ds +f(t),
\end{align}
where the memory kernel $K(t)$ can be obtained, e.g. through the second-order moment (time autocorrelation function) of $X(t)$ which can be obtained via Kubo's formula \eqref{intro_heat_linear_response}. However, the above equation is still unclosed since the dynamics of $X(t)$ depends also on the fluctuation force $f(t)$. This is resolved by the generalized second FDR proved in Section \ref{sec:2nd_FDR} which relates $M(t)$ and $f(t)$ as:
\begin{align}\label{intro-general_2nd_FDT}
M_{ij}(t)=\sum_{k=1}^2 G^{-1}_{jk}\left(-\langle f_{k}(0), f_i(t)\rangle_{\rho}+\langle \kappa_{k}^2\partial_{X_k}(
\ln\rho), f_i(t)\rangle_{\rho}\right),\qquad 1\leq i,j\leq 2,
\end{align}
where the meaning of each term will be explained later in Section \ref{sec:2nd_FDR}. As a comparison, the GLE \eqref{intro_gle} can be viewed as a reduced-order model based on the {\em first principle} since it is {\em derived} from the underlying (S)NS equation, while the well-established PDF method by Pope \cite{pope2001turbulent} is a phenomenological modeling of $X(t)$ using the random field assumption for the turbulent flow.

\section{Statistical mechanics for turbulence}
\label{sec:Noneq_turbulence}
\subsection{Equilibrium statistical mechanics for inviscid fluids} 
\label{sec:Equi_turbulence}
The classical statistical mechanics considers a $2N$-dimensional Hamiltonian system with canonical coordinate $\{p, q\}$. Such system satisfies the Hamiltonian equation of motion:
\begin{align}\label{eqn:can_Hamiltonian}
    \frac{dr}{dt}=[r,H]=J\nabla_{r}H,
\end{align}
where $r=\{p,q\}$, $H=H(p,q)$ is the Hamiltonian of the system and $[\cdot,\cdot]$ is the Poisson bracket defined as
\begin{align*}
[f,g]=\sum_{i,j=1}^N\frac{\partial f}{\partial{r_i}}J_{ij}\frac{\partial g}{\partial{r_j}},
\qquad \text{where $J$}=\left[\begin{matrix}
  0 & I_N\\ 
  -I_N & 0
\end{matrix}\right] \quad \text{is symplectic}.
\end{align*}
The Hamiltonian equation \eqref{eqn:can_Hamiltonian} implies the detailed Liouville theorem for the canonical system:
\begin{align}\label{eqn:Liouville_thm_can}
    \sum_{i=1}^N(\partial \dot{q}_i/\partial q_i+\partial \dot{p}_i/\partial q_i)=0.
\end{align}
From which, one can show that the phase space volume $V=\int dpdq$ is invariant. The classical statistical mechanics are concerned with the conserved quantities $\{C_1(p,q),\cdots C_n(p,q)\}$ of the Hamiltonian system since the probability measure of the form: 
\begin{align}\label{Ham:invariant_mea}
    \mu=\frac{1}{Z}F(C_1,\cdots,C_n)dpdq
\end{align}
is a invariant measure of the system, where $F$ is a smooth function and $Z$ is the normalization constant known as the partition function. With the ergodicity assumption, we can calculate the averaged value of observables by evaluating the ensemble average with respect to the invariant measure \eqref{Ham:invariant_mea}. There are three ensembles of particular importance: the microcanonical ensemble, the canonical ensemble and the grand canonical ensemble. Here we only review the first two because they are most relevant with the following context. For isolated Hamiltonian system, all conserved quantities $\{C_1(p,q),\cdots C_n(p,q)\}$ stays in a thin shell of the phase space, therefore the microcanonical ensemble:
\begin{align}\label{invar_mirco}
    \mu_m=\frac{1}{Z_{m}}\prod_{i=1}^n\delta(C_i(p,q)-C_i)dpdq
\end{align}
is suitable for the study of statistics of observables. For closed systems interacting with a large thermal bath with a fixed temperature $T$, we use the canonical ensemble of the form:
\begin{align}\label{invar_cano}
    \mu_c=\frac{1}{Z_{c}}\exp\left\{-\beta\sum_{i=1}^nC_i(p,q)\right\}dpdq,
\end{align}
where $\beta\propto1/T$ and $T$ is known as the equilibrium temperature. In the most common scenarios, only the conserved energy $C_1=H(p,q)$ is included in the probability measure, which leads to the microcanonical ensemble $\delta(H(p,q)-E)/Z_m$ and canonical ensemble $e^{-\beta H(p,q)}/Z_c$. 

Many inviscid fluid systems such as the KdV equation, inviscid Burgers' and Euler equation are Hamiltonian systems. However, they are fundamentally different from Hamiltonian system \eqref{eqn:can_Hamiltonian} in two aspects. First, these systems are infinite-dimensional. Secondly, fluid Hamiltonian systems normally do not have canonical structure. To study the Hamiltonian properties of inviscid fluids, we are lead to introduce a infinite-dimensional, non-canonical Poisson bracket $\{\cdot,\cdot\}$ defined as \cite{morrison1998hamiltonian,bouchet2010invariant}:
\begin{align*}
    \{F,G\}=\int\sum_{i,j}\frac{\delta F}{\delta\phi_i}\J_{ij}\frac{\delta G}{\delta\phi_j}d\mu,
\end{align*}
where $F=F(\phi),G=G(\phi)$ are functional of a field $\phi=\phi(\mu,t)$. Here $\mu=(\mu_1,\cdots,\mu_n)$ is a Eulerian observation variable and $\J$ is a cosympletic operator. As a specific example, consider the Euler equation for the incompressible, inviscid fluid:
\begin{align}\label{eqn:Euler}
\begin{dcases}
    \partial_t u+u\cdot\nabla u+\nabla p=0\\
    \nabla\cdot u=0.
\end{dcases}
\end{align}
In \eqref{eqn:Euler}, $u=u(x,t)$ is the Eulerian velocity at time $t$ and position $x$, $p=p(x,t)$ denotes the pressure. For two-dimensional Euler equation, it is more convenient to work with the equivalent dynamics for the vorticity $\omega=\nabla\land u=\partial_{x_2}u_1-\partial_{x_1}u_2$. By imposing the periodic boundary condition in a torus $\T^2=[-\pi,\pi]\times[-\pi,\pi]$, \eqref{eqn:Euler} can be reformulated as 
\begin{align}\label{eqn:Euler_w}
\partial_t\omega+(u\cdot\nabla) \omega=0.
\end{align}
With these settings, Euler equation \eqref{eqn:Euler_w} can be written as a Hamiltonian equation of motion:
\begin{align*}
    \partial_t\omega=\{\omega,H\},\qquad \text{where $H$}= \frac{1}{2}\int_{\T^2}u^2dx,
\end{align*}
and the non-canonical Poisson bracket has explicit expression:
\begin{align*}
    \{F,G\}=\int_{\T^2}\omega\left[\frac{\delta F}{\delta\omega},\frac{\delta G}{\delta\omega}\right]dx, \qquad \text{with} \ \J=-[\omega,\cdot],
\end{align*}
where 
\begin{align*}
   [f,g]=\frac{\partial f}{\partial x_1}\frac{\partial g}{\partial x_2}-\frac{\partial f}{\partial x_2}\frac{\partial g}{\partial x_1}.
\end{align*}
For non-canonical Hamiltonian systems, the phase volume conservation result does not hold in general \cite{morrison1998hamiltonian,kraichnan1980two}. Mathematically speaking, the non-canonical Hamiltonian system has a degenerate Poisson structure instead of a symplectic structure \cite{morrison1998hamiltonian,arnold1999topological}. The degeneracy of the Poisson bracket implies that any phase space functional which satisfies $\{C,H\}$ is a conserved quantity. For inviscid fluid systems, this leads {\em additional} invariant functionals $\{C_i\}$ which are called as the Casimir invariants. For 2D Euler equation \eqref{eqn:Euler}, we have infinite many Casimir invariants which are given by
\begin{align*}
    C[\omega]=\int_{\T^2} F(\omega)dx,
\end{align*}
where $F(\omega)$ is an arbitrary smooth function of $\omega$. Common choice of $F$ for the statistical mechanics study is $F=\omega ^2$ which corresponds to the enstrophy. For other inviscid fluid systems such as the 1D compressible fluid equation:
\begin{equation}
\begin{aligned}
\begin{dcases}
    \partial_tu&=-u\partial_xu-\frac{1}{\rho}\partial_xp\\
    \partial_t\rho&=-\partial_x(\rho u)
\end{dcases}
\end{aligned}
\end{equation}
in $\T^2$. We can only find two Casimir invariant $C_1[u]=\int _{\T^2}udx$ and $C_2[u]=\int_{\T^2} \rho dx$. With the Hamiltonian and the additional Casimir invariants, we can define analogously the microcanonical and the canonical measure for inviscid fluids. As an example, for 2D Euler equation, we have:
\begin{align}
    \mu_{m}&=\frac{1}{Z_{m}}\delta(H[\omega]-H)\prod_{p=1}^{\infty}\delta(G_{2p}[\omega]-G_{2p})D[\omega] \qquad\qquad (\text{Microcanonical measure})\label{inv:Euler_micro}\\
     \mu_{c}&=\frac{1}{Z_{c}}\exp\left\{-\beta_0 H[\omega]+\sum_{k=1}^{\infty}\beta_{p}G_{2p}[\omega]\right\}D[\omega] \qquad\qquad (\text{Canonical measure})\label{inv:Euler_cano},
\end{align}
where $D[\omega]$ is a functional differential, $G_{2p}[\omega]=\int_{\T^2}\omega^{2p}dx$ are the Casimir invariants. $\{\beta_0,\beta_p\}\propto\{\frac{1}{T_{H}},\frac{1}{T_{G_{2p}}}\}$ and $\{\frac{1}{T_{H}},\frac{1}{T_{G_{2p}}}\}$ can be understood as the "generalized" temperature. Here we only consider countable many $G_{2p}[\omega]$. This is because the exact statistical equilibrium for the Euler equation is normally assumed to be a even function $F$. Using Taylor expansion, it can be equivalently represented using countable many $G_{2p}[\omega]$. In fact, with only one Casimir invariant $G_2[\omega]$, i.e. the enstrophy, the resulting energy-enstrophy measure defined as \eqref{inv:Euler_micro}-\eqref{inv:Euler_cano} would yield satisfactory approximations of the exact equilibrium state  \cite{kraichnan1980two}.  

The canonical equilibrium measure defined as \eqref{inv:Euler_cano} is mathematically pathological, i.e. it is impossible to find finite value of $\{\beta,\alpha_p\}$ such that the system's mean energy is finite. Similar phenomenon was observed in the field theory \cite{parisi1988statistical,Justin} and we can use a technique called regularization to remove the infinity. The regularization we will apply simply refer to the ultraviolet truncation of the velocity (or vorticity) field in the Fourier mode space. As a result, we can get the equilibrium measure for the finite system \footnote{The truncated energy $H=H[\omega_i]$ and the Casimir invariants $G_{2p}=G_{2p}[\omega_i]$ must be conserved for the truncated system.}:
\begin{align}
    \mu_{m,N}&=\frac{1}{Z_{m,N}}\delta(H[\omega_i]-H)\prod_{p=1}^{K}\delta(G_{2p}[\omega_i]-G_{2p})\prod_{i=1}^Nd\omega_i \qquad\quad (\text{Regularized microcanonical measure})\label{inv:Euler_micro_T}\\
     \mu_{c,N}&=\frac{1}{Z_{c,N}}\exp\left\{-\beta_N H[\omega_i]+\sum_{p=1}^{K}\alpha_{p,N}G_{2p}[\omega_i]\right\}\prod_{i=1}^Nd\omega_i \qquad\quad (\text{Regularized canonical measure})\label{inv:Euler_cano_T},
\end{align}
where $\omega_i$ is the $i$-th Fourier mode of $\omega$ and the parameter $\{\beta,\alpha_p\}$ is re-denoted as $\{\beta_{N},\alpha_{p,N}\}$. Similar treatments are applicable to the KdV equation, inviscid Burgers' equation and three-dimensional ideal fluid equation\cite{morrison1998hamiltonian}, and the equilibrium statistical mechanics setting follows immediately. At the current stage, there is no need to study the grand canonical ensemble if the {\em same} finite range of Fourier mode is under investigation. This is because the total number of Fourier modes, like the total number of ``partices'', will not change for commonly used numerical solvers of the equation.

With the well-defined statistical equilibrium measures, we can now invoke the turbulence thermalization assumption. This hypothesis is implicitly stated in Krachinan's classical paper \cite{kraichnan1980two} for 2D Euler equation. We can generalize it to arbitrary invisid fluids and the result can be stated as: 
\begin{assumption}\label{assumption:therm} (Turbulence thermalization) The inviscid flow such as the 2D Euler equation and the KdV equation undergoes a thermalization process and reaches the classical statistical ensemble characterised by the invariants of the flow.
\end{assumption}
This assumption is a mixture of the ensemble assumption which states that the classical statistical mechanics ensembles, i.e. the microcanonical and canonical of the form \eqref{inv:Euler_micro}-\eqref{inv:Euler_cano} or \eqref{inv:Euler_micro_T}-\eqref{inv:Euler_cano_T}, describe the statistical properties of fluid observables, and an ergodicity assumption predicting that the flow will approach such equilibriums as $t\rightarrow+\infty$. Generally speaking, it is very hard to prove or disprove this assumption in a mathematically rigorous way. Since the assumption itself covers a large range of fluid phenomena, i.e, any inviscid fluids, we would suspect it is not hard to find numerical or theoretical evidence which is against the statement. Hence Assumption \ref{assumption:therm} should not be understood as a rigorously stated mathematical conjecture, but rather a guiding principle such as the ergodicity and thermalization assumption for canonical Hamiltonian systems. A more pragmatic point of view is: any inviscid fluid equation is a idealized model for real fluid phenomenon. Dissipation (additional terms such as $\nu\Delta \omega$) exists for most realistic fluids. These dissipative forces will drive the turbulence away from the equilibrium. Hence the meaning of Assumption \ref{assumption:therm}, as well as the validity of the putative equilibrium measure \eqref{inv:Euler_micro_T}-\eqref{inv:Euler_cano_T}, should be judged by their usage in predicting the behavior of realistic fluid system. With that being said, the discussion of the fluid Hamiltonian system and the identification of the explicit form of the ensemble distribution already makes it clear in the framework of statistical mechanics what means exactly when we say the real-world turbulence is a {\em nonequilibrium} phenomenon. 

\paragraph{\textbf{Remark}} For 2D Euler equation, we note that there are contradicting evidences with some supporting and some opposing Assumption \ref{assumption:therm}. On the one hand, Kraichnan used it to predict the existence of inverse energy cascade in turbulence which was numerically verified afterwards \cite{kraichnan1980two}. Some early numerical evidence for the thermalization of 2D turbulence according to the postulated equilibrium measure was summarized in \cite{kraichnan1980two}. On the other hand, Bouchet et al \cite{bouchet2010invariant} noticed the existence of Young measure which indicates a large potion of vorticity fields are dynamically invariant. Moreover, some numerical simulations \cite{chavanis1998classification} seem to observe the localization of the vorticity field which is a clear indication of non-ergodicity. However, the recent numerical studies \cite{krstulovic2009cascades,cichowlas2005effective} with  high-resolution direct numerical simulations (DNS) for the truncated 3D Euler equation provide strong evidence of thermalization according to the postulated canonical distribution \eqref{inv:Euler_cano_T}. 

\vspace{0.2cm}
Up to this point, we got a basic statistical mechanical framework for the inviscid fluid. In Table 1, we use the 2D Euler equation as an example and compare its statistical mechanical settings with commonly used one in condensed matters. With this framework, we can define or calculate many static equilibrium quantities. For example, the turbulence Boltzmann entropy for 2D Euler equation can be calculated as 
\begin{align*}
    S=\kappa_B\ln(Z_m) \qquad\text{or}\qquad S_N=\kappa_B\ln(Z_{m,N}),
\end{align*}
where $Z_m$ and $Z_{m,N}$ are the microcanonical partition functions. From the above discussion, one can see that the statistical mechanics for fluids have many features different from the classical statistical mechanics for condensed matters. Essentially because of the degeneracy of the Poisson structure, the settings for equilibrium measure  is more complex (Compare \eqref{invar_mirco}-\eqref{invar_cano} with the canonical system ensemble $\delta(H(p,q)-E)/Z_m$ and $e^{-\beta H(p,q)}/Z_c$). The ergodicity problem for inviscid fluids, if not harder, is as difficult as the one for condensed matters. There are many interesting open problems in the equilibrium statistical mechanics setting for fluids. For instance, how to select finite number of Casimir invariants if the fluid system has infinite many such that the resulting equilibrium measure fits the numerical results? Another interesting question regards the ensemble equivalence:  

\vspace{0.2cm}
\noindent\textbf{Problem 1.} {\em (Ensemble equivalence) For inviscid fluid systems, if the turbulence thermalization assumption holds, is there an equivalence between the canonical ensemble and the microcanonical ensemble in terms of describing the equilibrium statistical properties of the fluid?  If not, which one is closer to the reality?} 
\vspace{0.2cm}

For condensed matters in the thermodynamic limit, the answer to Problem 1 is normally yes. However, Bouchet et al.\cite{bouchet2010invariant} argued that for turbulent systems which have long-range interactions, this might not be the case. It is still unclear for general inviscid fluid systems such as the KdV equation, the inviscid Burgers' equation and three-dimensional ideal fluid, whether we have the ensemble equivalence. If not, how well the canonical ensemble approximates the microcanonical ensemble? 

\begin{table}[t]
\caption{Statistical mechanics for turbulence and the canonical Hamiltonian systems.}
\centering
\begin{tabular}{c|c|c|c|c|c}
\hline
\multicolumn{2}{c|}{} &\multicolumn{2}{c|}{Turbulence}&\multicolumn{2}{c}{Canonical Hamiltonian system}\\ 
\hline
\multirow{7}{*}{\STAB{\rotatebox[origin=c]{90}{Equilibrium}}}
&Equation of motion &\multicolumn{2}{c|}{Eqn  \eqref{eqn:Euler}}&\multicolumn{2}{c}{Eqn \eqref{eqn:can_Hamiltonian}}\\ 
&Dynamics structure&\multicolumn{2}{c|}{Poisson structure}&\multicolumn{2}{c}{Symplectic structure}\\ 
&Conserved quantity&\multicolumn{2}{c|}{$H[\omega]$, Casimir invariants $C[\omega]$}&\multicolumn{2}{c}{$H(p,q)$}\\ 
&Typical equilibrium measures &\multicolumn{2}{c|}{\eqref{inv:Euler_micro}-\eqref{inv:Euler_cano}}&\multicolumn{2}{c}{$\delta(H(p,q)-H)/Z_m$, $e^{-\beta H(p,q)}/Z_c$}\\
&Regularization &\multicolumn{2}{c|}{Required}&\multicolumn{2}{c}{Not required}\\ 
&Ensemble equivalence &\multicolumn{2}{c|}{Not sure}&\multicolumn{2}{c}{Normally yes if $N\rightarrow\infty$}\\ 
&Dynamical ergodicity &\multicolumn{2}{c|}{Not sure}&\multicolumn{2}{c}{Likely if not integrable}\\ 
\hline
\multirow{12}{*}{\STAB{\rotatebox[origin=c]{90}{Nonequilibrium}}}
& External term for NESS1 &\multicolumn{2}{c|}{$\nu\Delta \omega$}&\multicolumn{2}{c}{Gaussian thermostat}\\ 
& NESS1 measure &\multicolumn{2}{c|}{SRB measure(?)}&\multicolumn{2}{c}{SRB measure(?)}\\ 
& Dynamical ergodicity &\multicolumn{2}{c|}{Numerically observed}&\multicolumn{2}{c}{Numerically observed}\\ 
&External term for NESS2&\multicolumn{2}{c|}{$\nu\Delta \omega+f(t,\xi)$}&\multicolumn{2}{c}{Langevin thermostat}\\
& NESS2 measure &\multicolumn{2}{c|}{$\rho \prod_{i=1}^Nd\omega_i$, where $\K^*\rho=0$}&\multicolumn{2}{c}{$\rho \prod_{i=1}^Ndp_idq_i$, where $\K^*\rho=0$}\\ 
& Dynamical ergodicity &\multicolumn{2}{c|}{Proved conditionally in \cite{hairer2006ergodicity,weinan2001ergodicity,bedrossian2019almost}}&\multicolumn{2}{c}{Proved conditionally in \cite{rey2002exponential,cuneo2018non}}\\ 
&External term for NESS3&\multicolumn{2}{c|}{$\nu\Delta \omega+g(t)$}&\multicolumn{2}{c}{Nos\'e-Hoover thermostat}\\
& NESS3 measure &\multicolumn{2}{c|}{Unknown}&\multicolumn{2}{c}{Unknown}\\
& Dynamical ergodicity &\multicolumn{2}{c|}{Numerically observed}&\multicolumn{2}{c}{Numerically observed}\\
&Transport observable&\multicolumn{2}{c|}{$\Tr e^{\int_0^t\L(s,\xi)ds}X(0)$}&\multicolumn{2}{c}{$J(p(t),q(t))$}\\
&Near equilibrium&\multicolumn{2}{c|}{$\nu, f(t,\xi), g(t)\ll 1$}&\multicolumn{2}{c}{$\Delta T\ll 1$}\\ 
&Far-from equilibrium&\multicolumn{2}{c|}{$\nu, f(t,\xi), g(t)\geq o(1)$}&\multicolumn{2}{c}{$\Delta T\geq o(1)$}\\ 
\hline
\end{tabular}
\label{Tab:table}
\end{table}

\subsection{Nonequilibrium statistical mechanics for viscous fluids}\label{sec:Noneqn_viscous}
Different from canonical Hamiltonian systems, almost all realistic physics for turbulence happen in the nonequilibrium regime. This is due to the fact we just mentioned: dissipation and external forces exist in almost all realistic fluid systems, with only several exceptions such as the superfluid. Nonequilibrium phenomena are generally much more complex than the equilibrium ones simply because, as remarked by Ruelle \cite{ruelle2003there}: \say{Nonequilibrium covers a variety of different physical phenomena, such as
hysteresis and decay of metastable states, and one does not really expect a unified description for all those.} Of course for the nonequilibrium statistical mechanics of turbulence, there is no reason to suspect the situation would be different. Hence we would only focus on the study of nonequilibrium steady
states (NESSs) and we are mainly concerned with the transport phenomena. For turbulence, this already contains many similar cases with different mathematical interpretations, which we will detail below.

We begin our discussion with canonical Hamiltonian systems. The common way to numerically simulate a nonequilibrium system is by adding external forces to the Hamiltonian equation of motion \eqref{eqn:can_Hamiltonian}. These forces can be deterministic or stochastic and are called as {\em thermostats} \cite{morriss2013statistical}. They are many possible available thermostats and we only review three of them. In the following models, we assume the equal mass $m=1$ for all particles. 
\paragraph{\textbf{Gaussian thermostat}} The Gaussian thermostat for a canonical Hamiltonian systems has equation of motion:
\begin{equation}\label{eqn:G_thermo}
\begin{aligned}
\begin{dcases}
\frac{dq_i}{dt}&=p_i\\
\frac{dp_i}{dt}&=-\frac{\partial H}{\partial q_i}-\alpha_i p_i,
\end{dcases}
\end{aligned}
\end{equation}
where $-\alpha_i p_i$ is the external force added to the system and $\alpha=\alpha(p,q)$ is a time-independent function of space variable. A common choice of $\alpha$ is $\alpha=\partial_{q}H\cdot p/p\cdot p$ which yields the {\em Gaussian isokinetic thermostat}. The corresponding Gaussian isokinetic dynamics {\em does not} preserve equilibrium distribution such as the canonical distribution and only approximate it in the thermodynamical limit $N\rightarrow\infty$ \cite{morriss2013statistical}. A mathematically rigorous justification of the steady state generated by the Gaussian isokinetic dynamics has to use the Sinai-Ruelle-Bowen (SRB) measure \cite{ruelle1999smooth}, which also works for chaotic dissipative systems with attractors \cite{young2002srb}. By choosing $\alpha$ to be non-negative constants, this system is dissipative and we expect the NESS can still be described by the SRB measure. For such a case, Eqn \eqref{eqn:G_thermo} can be viewed as a "bare", deterministic heat conduction model when comparing with the following Langevin thermostat heat conduction model. 
\paragraph{\textbf{Langevin thermostat}} The Langevin thermostat adds stochastic forces to the Hamiltonian systems. A typical Langevin thermostat to simulate equilibrium systems is given by SDE:
\begin{equation}\label{eqn:SDE_langevin}
\begin{aligned}
\begin{dcases}
dq_i&=p_idt\\
dp_i&=-\partial_{q_i} Hdt-\gamma p_idt+\sigma d\W_i(t),
\end{dcases}
\end{aligned}
\end{equation}
where $\W_i(t)$ are independent Wiener processes. In particular, if the constant $\sigma$ satisfies the fluctuation-dissipation theorem $\sigma^2=2\gamma\kappa_BT_{eq}$, then the canonical distribution is the invariant distribution of the SDE \eqref{eqn:SDE_langevin}. Similar models can also be used to simulate nonequilibrium systems by connecting the boundary particles with Langevin thermostats of {\em different} temperatures. This leads to a stochastic heat conduction model \cite{lepri2003thermal,cuneo2018non,zhu2021hypoellipticity3} of the form:
\begin{equation}\label{SDE:n_d_heat}
\begin{dcases}
\begin{aligned}
dq_i&=p_idt\\
dp_i&=-\partial_{q_i}Hdt-\gamma_ip_idt+\sqrt{2k_BT_i\gamma_i}d\W_i(t)
\end{aligned}
\end{dcases}
\qquad\qquad 
i\in\G,
\end{equation}
where $\G$ is the graph of interacting particles (see a 2D example in Figure \ref{fig:Heat_transport}). In \eqref{SDE:n_d_heat}, $\gamma_i>0$ if $i\in\B$, where $\B$ is the boundary of graph $\G$ and $\gamma_i=0$ if $i\in\G\setminus\B$. $T_i$ is the temperature of the $i$-th thermostat. The NESS for the heat conduction model with Langevin thermostats corresponds to the invariant measure of SDE \eqref{SDE:n_d_heat}.
\paragraph{\textbf{Nos\'e-Hoover thermostat}} The Nos\'e-Hoover thermostat is a deterministic thermostat that preserves the canonical distribution. The dynamics is given by the Hamiltonian equation of motion with the extended-Hamiltonian:
\begin{equation}\label{Ham_NH}
\begin{aligned}
H=\sum_{i=1}^N\frac{p_i^2}{2s}+\phi(q)
+\frac{p_s^2}{2Q}+(g+1)\kappa_BT\ln s,
\end{aligned}
\end{equation}
where $s$ is an additional degree of freedom which has effect on each $p_i$. By connecting the boundary particles with Nos\'e-Hoover thermostats of {\em different} temperatures, we can get a deterministic heat conduction model \cite{lepri2003thermal,zhu2021hypoellipticity3}:
\begin{align}\label{Nose-Hover_heat}
\begin{dcases}
    \frac{dq_i}{dt}&=p_i\\ 
    \frac{dp_i}{dt}&=-\partial_{q_i}\H(p,q)-
    \begin{dcases}
    \gamma_L p_i,\qquad \text{if $i\in \B_{T_L}$}\\
    \gamma_R p_i,\qquad \text{if $i\in \B_{T_R}$},
    \end{dcases} 
\end{dcases}    
\end{align}
where $\B_{L,R}$ are the set of boundary oscillators which interact with thermostat at the temperature $T_{L,R}$. The cardinality of $\B_{L,R}$ are denoted as $|\B_{L,R}|$. The dynamics of the auxiliary variables $\gamma_{L,R}$ are given by:
\begin{align}\label{Nose-Hover_bath}
    \frac{d\gamma_{L,R}}{dt}=\frac{1}{\theta_{L,R}}\left(\frac{1}{k_BT_{L,R}|\B_{T_{L,R}}|}\sum_{n\in \B_{T_{L,R}}}p_n^2-1\right),
\end{align}
where constant $\theta_{L,R}$ is the thermostat response time. If the system is imposed with the nonequilibrium condition $T_L\neq T_R$, there is no well-defined mathematical measure found for the NESS corresponding to the Nos\'e-Hoover heat conduction model. 

\begin{figure}[t]
\centerline{
\includegraphics[height=6.5cm]{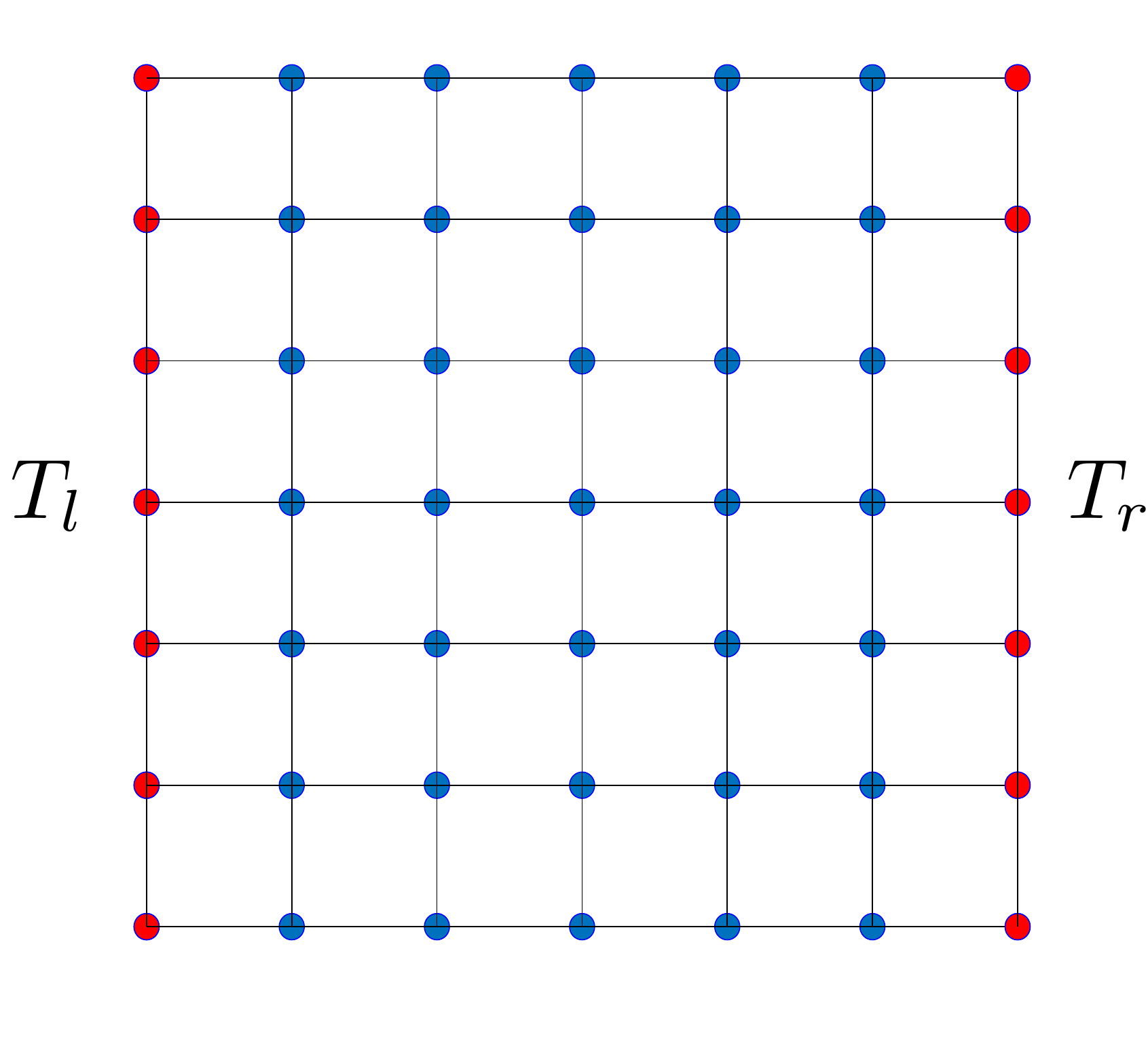}
\includegraphics[height=6.5cm]{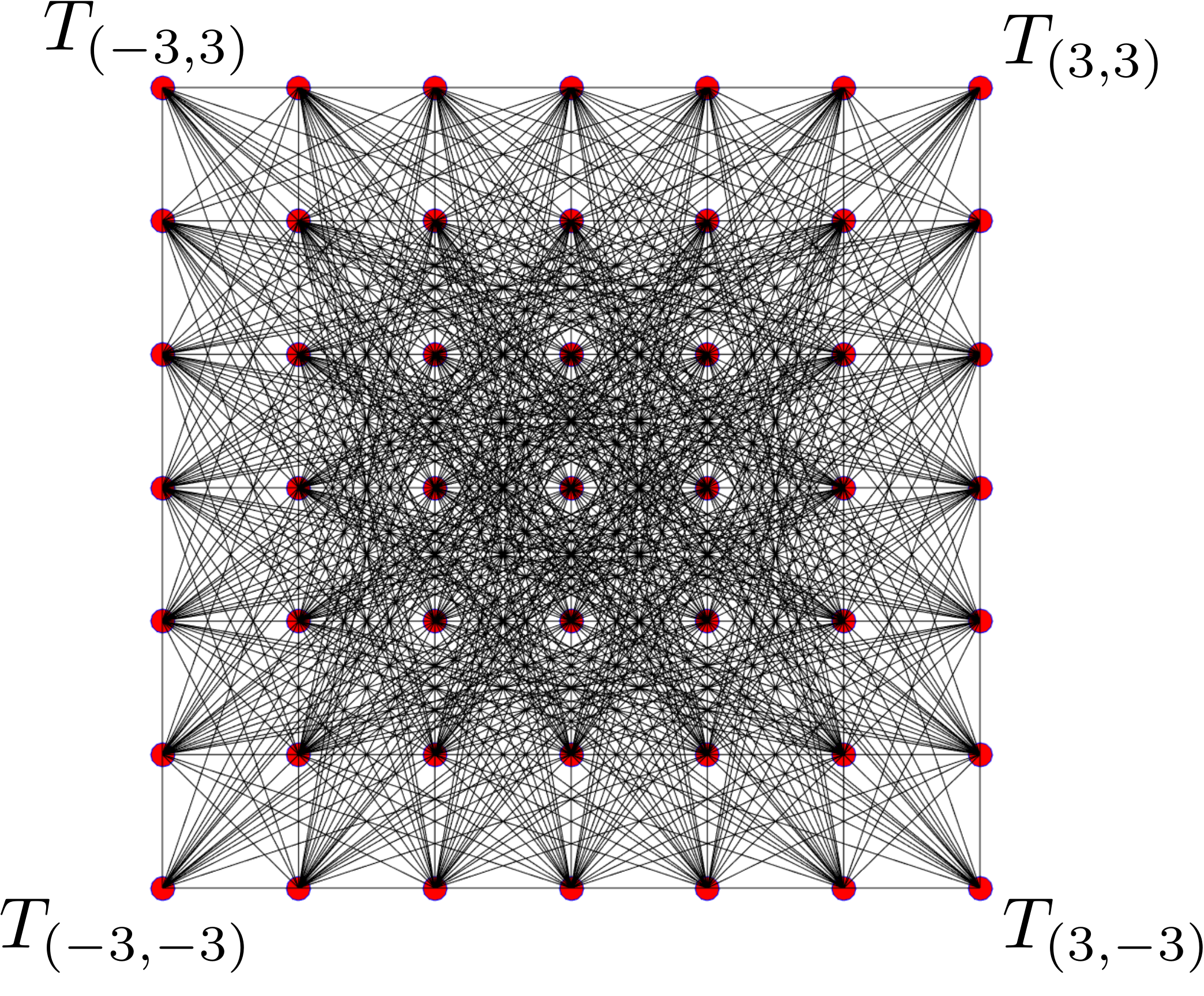}
}
\caption{Comparison of a 2D heat conduction model (Left) and the nonequilibrium statistical mechanics setting for the 2D viscous fluid (Right). For the heat conduction model, the left and the right boundary particles (marked in red) are connected with thermostats of temperature $T_l$ and $T_r$. Particles with interactions are linked with each other with black lines. For frequently used statistical mechanics model, the interaction is limited to the neighborhood of each particle hence the dynamics can be represented as a square lattice. For the viscous fluid system, the Fourier modes can be understood as the ``particle" which has long-range interactions (see expression \eqref{eqn:alpha_beta0}). Since the dissipative force $-\nu\Delta\omega$ acts on all modes, all particles in the lattice are marked in red. The effective ``temperature" are only displayed for the corner Fourier modes.}
\label{fig:Heat_transport} 
\end{figure}
\vspace{0.2cm}
These three heat conduction models make a good analogy of three commonly used settings for the simulation of viscous fluids. In the Fourier mode space, consider the fluid equation for a incompressible, viscous fluid with body force $F(x,t)$: 
\begin{align}\label{eqn:NS_Fourier}
\partial_tu_k=B(u,u)-\nu|k|^2u_k+F_k,
\end{align} 
where $F_k(t)$ is the $k$-th Fourier mode of $F(x,t)$. Then $\{u_k\}$ can be understood as the ``particle'' (see schematic illustration in Figure \ref{fig:Heat_transport}). Moreover, the NESSs created by different forces $F(x,t)$ and their similarities with the heat conduction model are be summarized as follows:
\paragraph{\textbf{NESS1}} If $\nu>0$ and $F(x,t)=0$, then the viscous term $\nu\Delta\omega$ added linear dissipation force $-\nu|k|^2\omega_k$ in Fourier mode $\omega_k$, similar to the Gaussian thermostat heat conduction model. We expect the NESS of the fluid equation under this setting, named as NESS1, can be characterised by the SRB measure, although an exact construction of such measure for turbulence seems to be extremely difficult \cite{ruelle2012hydrodynamic,gallavotti1995dynamical}.  
\paragraph{\textbf{NESS2}} If $\nu>0$ and $F(x,t)=f(t,\xi)$ is a white noise imposed in selected Fourier modes (see details in Section \ref{sec:Turb_disper_setting}), then this fluid model can be understood as a ``Langevin thermostat heat conduction model". The NESS of the fluid equation under this setting, named as NESS2, has distribution $\rho$ which is the solution of the stationary Fokker-Planck equation $\K^*\rho=0$.
\paragraph{\textbf{NESS3}} If $\nu>0$ and $F(x,t)=g(x,t)$, where $g(x,t)$ is a deterministic force added in the Fourier modes, then this fluid model can be understood as a ``Nos\'e-Hoover thermostat heat conduction model". The NESS of the fluid equation under this setting, named as NESS3, has not been understood clearly.

\vspace{0.2cm}
The correspondence between the heat conduction model and the fluid mechanics from the perspective of nonequilibrium statistical mechanics was noticed by Ruelle \cite{ruelle2012hydrodynamic,ruelle2014non}. The formulation we presented is a re-interpretation of Ruelle's original idea \cite{ruelle2012hydrodynamic} in the Fourier mode space. However, different from \cite{ruelle2012hydrodynamic,ruelle2014non}, in which Ruelle focused more on the {\em static} properties of the NESS, the main theme of this paper is about the {\em dynamics} of observables in the NESS since we will are mainly concerned with the transport phenomena. Here we note more similarities between the heat conduction model and the fluid equations. The following facts indicate that this analogy is more than a pure phenomenological observation. In particular, we emphasize Point 3 which implies that by connecting the fluid system with the heat conduction model, we can obtain physically meaning results to describe the turbulent dispersion.
\begin{enumerate}
    
    \item The local temperature at each lattice grid of the heat conduction model has a microscopic definition as the averaged kinetic energy of the particle \cite{lepri2003thermal}. Corresponding, the local ``temperature" for each Fourier mode can be understood as the averaged kinetic energy $\langle|u_k|^2\rangle$, and the energy spectrum $E(k)$ of turbulence corresponds to the temperature profile $T(x)$ \cite{lepri2003thermal} of the heat conduction model.
    
    \item For fully developed turbulent flow with small viscosity $\nu$, as the wavenumber $k$ increases, the energy spectrum $E(k)$ can be roughly divided into the {\em source range}, the {\em inertial range} and the {\em viscous dissipative range}. The dominant energy process at low wavenumbers is the transfer of kinetic energy between Fourier modes and the dissipation is negligible. Hence an effective graph representation for highly turbulent flows would be as the physical picture of the real world heat conduction, i.e. the left plot in Figure \ref{fig:Heat_transport}, where the dissipative forces only act at the boundary or near-boundary particles. 
        
    \item In Section \ref{sec:FDT}, we show that the generalized fluctuation-dissipation relations for the heat conduction model \cite{zhu2021hypoellipticity3} can be directly applied to the fluid model. In particular, we derive heuristically a path-integral-form linear response formula for the fluid system by considering it as a heat conduction model. This formula is rigorously proved for equilibrium and nonequilibrium systems. 
\end{enumerate}

We also have many open problems for the NESSs of turbulence. In particular, the following one is interesting:

\vspace{0.2cm}
\noindent\textbf{Problem 2.} {\em (NESS equivalence) For the equation of viscous fluids with different external force terms, is there an equivalence between the resulting NESSs in terms of describing the nonequilibrium statistical properties of the fluid? If not, which one is closer to the reality?}
\vspace{0.2cm}

\noindent As we have seen from Table 1, the NESSs corresponding to different external forces has distinctive mathematical structures. Hence a general answer to the first question is no. Our interest is more of the numerical sense. For instance, if two NESSs generated by different external terms, say the viscous dissipation+white noise (NESS2) and viscous dissipation+Kolmogorov force (NESS3), shares the same energy spectrum, is there any dynamical similarities between turbulent transport phenomena such as the particle dispersion and the eddy dissipation? This type of questions are important because if such similarities exist, then one would expect the result obtained for stochastic fluid systems (NESS2) can be extended to address the more difficult deterministic turbulence case (NESS1 and NESS3).

\section{Turbulent dispersion as a problem of nonequilibrium transport}\label{sec:Turb_disper_setting}
Consider a two-dimensional SNS equation for the incompressible, viscous fluid with stochastic forcing $F=F(x,t)$: 
\begin{align}\label{eqn:NS}
\begin{dcases}
\partial_tu+(u\cdot\nabla) u=-\nabla p+\nu\Delta u+F,\\
\nabla\cdot u=0.
\end{dcases}
\end{align}
Similar to the Euler equation \eqref{eqn:Euler} considered in Section \ref{sec:Equi_turbulence}, we consider Eqn \eqref{eqn:NS} with periodic boundary conditions in a torus $\T^2=[-\pi,\pi]\times[-\pi,\pi]$. Working with the equivalent dynamics for the vorticity field $\omega=\nabla\land u=\partial_{x_2}u_1-\partial_{x_1}u_2$, \eqref{eqn:NS} can be reformulated as 
\begin{align}\label{eqn:NS_w}
\partial_t\omega+(u\cdot\nabla) \omega=\nu\Delta \omega+f,
\end{align}
where the pressure term is eliminated via integration and $f=\nabla \land F$. The external force term $f(x,t)$ is a Gaussian noise white in time and correlated in space which will be specified later. In this paper, we are mainly concerned with the transport of Lagrangian particles by this SNS equation. To this end, consider a characteristic Lagrangian particles $X(t)=X(t; x_0)$ with initial condition $X(0; x_0)=x_0$, the particle trajectory can be described by the equation of motion:
\begin{align}\label{eqn:Lag_EoM}
    dX_i=u_i(X(t),t)dt+\kappa_i d\D_i(t), \qquad i=1,2,
\end{align}
where $u(X(t),t)$ is the Eulerian velocity of the fluid at location $X(t)$, $\kappa_i$ is the molecular diffusivity and $\D_i(t)$ is a standard Wiener process. If $\kappa_i$ is chosen to be 0, then the Lagrangian particle is purely passive, i.e. moves along the flow of the fluid. The description of the turbulent dispersion via equation \eqref{eqn:Lag_EoM} takes a Lagrangian point of the view of the transport. Corresponding, the Eulerian description considers the turbulent transport of a space-time 
particle concentration function $\phi=\phi(x,t)$ where the evolution of $\phi$ is given by an advection-diffusion equation (assuming $\kappa_1=\kappa_2=\kappa$):
\begin{align}\label{eqn:Euler_EoM}
\begin{dcases}
 \partial_t\phi=u\cdot \nabla \phi+\kappa^2\Delta\phi\\
 \phi(0)=\varphi(x),
\end{dcases}
\end{align}
where $u=u(x,t)$ is the Eulerian velocity of the fluid. Mathematically speaking, these two descriptions are equivalent. In fact, \eqref{eqn:Euler_EoM} is the Kolmogorov backward equation corresponding to SDE \eqref{eqn:Lag_EoM} and $\phi(x,t)=\mathbb{E}_{D_i}[\varphi(X_t)|X_0=x]$. Numerically, one can solve the combined system  \eqref{eqn:NS}-\eqref{eqn:Euler_EoM} for $\phi(x,t)$, or equivalently, solve equations \eqref{eqn:NS}-\eqref{eqn:Lag_EoM} with {\em different} initial conditions $X(0,x_0)$ and construct $\phi(x,t)$ {\em a posterior} by evaluating the Feynman-Kac formula $\phi(x,t)=\mathbb{E}_{D_i}[\varphi(X_t)|X_0=x]$. For this work, we will use the Lagrangian description \eqref{eqn:Lag_EoM} to study turbulent dispersion.
\paragraph{\textbf{Remark 1}} With only the SDE \eqref{eqn:Lag_EoM}, one cannot write its corresponding Kolmogorov backward equation as \eqref{eqn:Euler_EoM} since Eqn \eqref{eqn:Lag_EoM} alone is not a time-homogeneous system \cite{pavliotis2014stochastic}. But when combined with the SNS equation \eqref{eqn:NS}, the extended system \eqref{eqn:Lag_EoM}-\eqref{eqn:NS} is time-homogeneous (see explicit expression below), therefore we get the correspondence between the Lagrangian description \eqref{eqn:Lag_EoM} and the Eulerian description \eqref{eqn:Euler_EoM}. A good analogy which helps to understand this is the Langevin dynamics \eqref{eqn:SDE_langevin} for particle systems, the SDE for $p_i$ is not time-homogeneous since it depends on $q_i(t)$. But the combined system \eqref{eqn:SDE_langevin} is time-homogeneous for both $p_i$ and $q_i$.
\paragraph{\textbf{Remark 2}} Another model for a small spherical aerosol particle in turbulence was introduced by Maxey in \cite{maxey1987gravitational}. The equation of motion for such a particle of radius $a$ and mass $m_p$ is given by
\begin{align}\label{eqn:Lag_Lagevin}
\begin{dcases}
    &dX_i=V_idt\\
    &m_pdV_i=6\pi a\nu(u_i(X(t),t)-V_i(t))dt+m_pg_idt+\xi_idD_i(t),
\end{dcases}
\end{align}
where $X(t),V(t)$ are the position and velocity of the particle, $u(x,t)$ is the Eulerian fluid velocity, $g$ is the gravitational acceleration and $\nu$ is the fluid viscosity. In order to model the thermal fluctuation of a Brownian particle, here we added a white noise term $\xi_idD_i(t)$ as in \eqref{eqn:Lag_EoM}, where $\xi$ is the fluctuation intensity. This equation represents a balance between the particle gravity, the Stokes-law drag force which exists due to velocity difference between the particle and the surrounding fluid and the molecular diffusion. In fact, \eqref{eqn:Lag_Lagevin} is a Langevin dynamics \eqref{eqn:SDE_langevin} for the aerosol particle with acceleration $\gamma(u_i(X(t),t)+m_pg_i/(6\pi a\nu))$, where $\gamma$ and $\xi_i$ satisfy $\gamma=6\pi a\nu/m_p$ and $\xi_i=\sigma m_p$. Under this setting, one can interpret Eqn \eqref{eqn:Lag_EoM} as the ovderdamped Langevin dynamics for \eqref{eqn:Lag_Lagevin} in the massless limit $m_p\rightarrow 0$.
We emphasize that all theoretical results obtained in this paper can be readily extended to the Lagrangian particle model \eqref{eqn:Lag_Lagevin}. 

\vspace{0.3cm}
Now we consider the stochastic force $f$ of the following form:
\begin{align}
    f=\sum_{k\in K}\sigma_k\cos(k\cdot x)d\B_k(t)+\gamma_k\sin(k\cdot x)d\W_k(t),
\end{align}
where $\B_k(t)$ and $\W_k(t)$ are independent standard Wiener processes which are also independent with $\D(t)$. $\sigma_k,\gamma_k$ are positive diffusion constants. The finite set $K\subset \Z^2\setminus\{0,0\}$ gives the forced mode which excludes the mean flow force with $k=(0,0)$. Due to the periodic BCs, we have the Fourier mode representation for the vorticity field $\omega(x,t)=\sum_k\alpha_{k}(t)\cos(k\cdot x)+\beta_k(t)\sin(k\cdot x)$, where $k\in \Z^+\times \Z^+$. Then Eqn \eqref{eqn:NS_w} can be rewritten as:
\begin{equation}\label{eqn:alpha_beta0}
\begin{aligned}
d\alpha_l&=\left\{-\nu|l|^2\alpha_{l}+\sum_{j+k=l}\frac{j^{\perp}\cdot k}{|j|^2}(-\alpha_j\alpha_k+\beta_j\beta_k)+\sum_{j-k=l}\frac{j^{\perp}\cdot k}{|j|^2}(\alpha_j\alpha_k+\beta_j\beta_k)\right\}dt+\sigma_{l}d\B_l\\
d\beta_l&=\left\{-\nu|l|^2\beta_{l}+\sum_{j+k=l}\frac{j^{\perp}\cdot k}{|j|^2}(\alpha_j\beta_k+\beta_j\alpha_k)+\sum_{j-k=l}\frac{j^{\perp}\cdot k}{|j|^2}(-\alpha_j\beta_k+\beta_j\alpha_k)\right\}dt+\gamma_{l}d\W_l,
\end{aligned}
\end{equation}
where $j^{\perp}=(j_1,j_2)^{\perp}=(-j_2,j_1)$ and $j,l\in \Z^+\times \Z^+$. Using the stream function $\psi=\psi(x,t)$ relations: $u=(\partial_{2}\psi(x,t),-\partial_{1}\psi(x,t))$ and $\omega(x,t)=-\Delta \psi(x,t)$, we can get the velocity field representation 
\begin{equation}\label{eqn:u1_u2}
\begin{aligned}
u_1(x,t)&=\sum_k\frac{\beta_kk_2}{|k|^2}\cos(k\cdot x)-\frac{\alpha_kk_2}{|k|^2}\sin(k\cdot x)\\
u_2(x,t)&=\sum_k\frac{\alpha_kk_1}{|k|^2}\sin(k\cdot x)-\frac{\beta_kk_1}{|k|^2}\cos(k\cdot x).
\end{aligned}
\end{equation}
In this paper, we consider a finite-dimensional Fourier mode approximation of exact solution $\omega(x,t)\approx\omega^N(x,t)$. By introducing a finite rank projection operator $P_N$ and applying it to \eqref{eqn:NS_w}, we can get the projected equation:
\begin{align}\label{eqn:NS_w_Truncated}
    \partial_t\omega^N+P_N(u^N\cdot\nabla)\omega^{N}=\nu\Delta\omega^N+f^N,
\end{align}
where $P_N$ eliminates all the Fourier modes bigger than the threshold $N$, i.e. for any index $|k|_{\infty}>N$. Similarly we can project the velocity field \eqref{eqn:u1_u2} into finite Fourier modes and get approximation $u(x,t)\approx u^N(x,t)$. In the Fourier modes space, the dynamics for the truncated NS equation \eqref{eqn:NS_w_Truncated} and the Lagrangian particle equation \eqref{eqn:Lag_EoM} can therefore be represented as 
\begin{equation}\label{eqn:alpha_beta}
\begin{aligned}
d\alpha_l&=\left\{-\nu|l|^2\alpha_{l}+\sum_{j+k=l}^N\frac{j^{\perp}\cdot k}{|j|^2}(-\alpha_j\alpha_k+\beta_j\beta_k)+\sum_{j-k=l}^N\frac{j^{\perp}\cdot k}{|j|^2}(\alpha_j\alpha_k+\beta_j\beta_k)\right\}dt+\sigma_{l}d\B_l\\
d\beta_l&=\left\{-\nu|l|^2\beta_{l}+\sum_{j+k=l}^N\frac{j^{\perp}\cdot k}{|j|^2}(\alpha_j\beta_k+\beta_j\alpha_k)+\sum_{j-k=l}^N\frac{j^{\perp}\cdot k}{|j|^2}(-\alpha_j\beta_k+\beta_j\alpha_k)\right\}dt+\gamma_{l}d\W_l,\\
dX_1&=\left\{\sum_k^N\frac{\beta_kk_2}{|k|^2}\cos(k\cdot X)-\frac{\alpha_kk_2}{|k|^2}\sin(k\cdot X)\right\}dt+\kappa_1d\D_1(t)\\
dX_2&=\left\{\sum_k^N\frac{\alpha_kk_1}{|k|^2}\sin(k\cdot X)-\frac{\beta_kk_1}{|k|^2}\cos(k\cdot X)
\right\}dt+\kappa_2d\D_2(t).
\end{aligned}
\end{equation}
We introduce the following simplified notation: 
\begin{equation}\label{eqn:alpha_beta_notation}
\begin{aligned}
F_l(\{\alpha_l\},\{\beta_l\})&=\sum_{j+k=l}^N\frac{j^{\perp}\cdot k}{|j|^2}(-\alpha_j\alpha_k+\beta_j\beta_k)+\sum_{j-k=l}^N\frac{j^{\perp}\cdot k}{|j|^2}(\alpha_j\alpha_k+\beta_j\beta_k)\\
G_l(\{\alpha_l\},\{\beta_l\})&=\sum_{j+k=l}^N\frac{j^{\perp}\cdot k}{|j|^2}(\alpha_j\beta_k+\beta_j\alpha_k)+\sum_{j-k=l}^N\frac{j^{\perp}\cdot k}{|j|^2}(-\alpha_j\beta_k+\beta_j\alpha_k),\\
\partial_{X_2}\psi(X,t)&=u_1(X(t),t)=\sum_k^N\frac{\beta_kk_2}{|k|^2}\cos(k\cdot X)-\frac{\alpha_kk_2}{|k|^2}\sin(k\cdot X)\\
-\partial_{X_1}\psi(X,t)&=u_2(X(t),t)=\sum_k^N\frac{\alpha_kk_1}{|k|^2}\sin(k\cdot X)-\frac{\beta_kk_1}{|k|^2}\cos(k\cdot X),
\end{aligned}
\end{equation}
where $\psi(x,t)$ is the stream function. Then Eqn \eqref{eqn:alpha_beta} can be rewritten as 
\begin{align}\label{eqn:alpha_beta_simplified}
    \begin{dcases}
    d\alpha_l&=F_{l}(\{\alpha_l\},\{\beta_l\})dt-\nu|l|^2\alpha_ldt+\sigma_ld\B_l\\
    d\beta_l&=G_{l}(\{\alpha_l\},\{\beta_l\})dt-\nu|l|^2\beta_ldt+\gamma_ld\W_l\\
    dX_1&=\partial_{X_2}\psi(X,t)dt+\kappa_1d\D_1(t)\\
    dX_2&=-\partial_{X_1}\psi(X,t)dt+\kappa_2d\D_2(t)
    \end{dcases}.
\end{align}
The above equation defines a diffusion process in $\R^{(N+1)^2+1}$. This allows us to define a composition operator $\M(t,0)$ that pushes forward
in time the average of the observable $O(t)=O(\{\alpha_l(t)\},\{\beta_l(t)\},X_1(t),X_2(t))$ over the noise, i.e., 
\begin{align}
\mathbb{E}_{\B_l,\W_l,\D_1,\D_2}[O(t)|\{\alpha_l(0)\},\{\beta_l(0)\},X_1(0),X_2(0)]= 
\M(t,0) O(0)=e^{t\K}O(\{\alpha_l(0)\},\{\beta_l(0)\},X_1(0),X_2(0)),
\label{MarkovSemi}
\end{align}
where $\M(t,0)=e^{t\K}$ is a Markovian semigroup, and its infinitesimal generator is known as the Kolmogorov backward operator $\K$, which can be written as 
\begin{align}\label{op:K_NS_extend}
    \K=\X_0+\frac{1}{2}\sum_{j}\X_j^2,
\end{align}
where 
\begin{equation}\label{eqn:X_0}
\begin{aligned}
    \X_0=\sum_{l}^N&\{-\nu|l|^2\alpha_{l}+F_l(\{\alpha_l\},\{\beta_l\})\}\frac{\partial}{\partial \alpha_l}+\sum_{l}^N\{-\nu|l|^2\beta_{l}+G_l(\{\alpha_l\},\{\beta_l\})\}\frac{\partial}{\partial \beta_l}\\
    &+\partial_{X_2}\psi(X(t),t)\frac{\partial}{\partial X_1}-\partial_{X_1}\psi(X(t),t)\frac{\partial}{\partial X_2}
\end{aligned}
\end{equation}
and 
\begin{equation}\label{eqn:X_k*X_k}
\begin{aligned}
\frac{1}{2}\sum_{j}\X_j^2
=\frac{1}{2}\sum_{k\in K}^N\sigma_k^2\left(\frac{\partial}{\partial{\alpha_k}}\right)^2
+\frac{1}{2}\sum_{k\in K}^N\gamma_k^2\left(\frac{\partial}{\partial{\beta_k}}\right)^2
+\frac{1}{2}\kappa_1^2\left(\frac{\partial}{\partial{X_1}}\right)^2
+\frac{1}{2}\kappa_2^2\left(\frac{\partial}{\partial{X_2}}\right)^2
\end{aligned}.
\end{equation}
Equation \eqref{eqn:alpha_beta} is the model we will use to study turbulence dispersion. In canonical nonequilibrium systems, any transport process is characterized by the corresponding flux which is a time-independent function of phase variables. Takes the heat conduction model as an example, the averaged heat flux $J(p(t),q(t))$ defined as \cite{lepri2003thermal,zhu2021hypoellipticity3}:
\begin{align}\label{def:local_global_flux}
J(p(t),q(t))=\frac{1}{N}\sum_{j}p_iV'(q_{j+1}-q_j),
\end{align}
where $V(q_{i+1}-q_i)$ is the interactive potential between neighbourhood particles, characterises the intensity of the heat transport. The time correlation function of $J$ can be used to calculate heat conductivity via Green-Kubo formula in the near equilibrium regime. The observable that quantifies the Lagrangian particle dispersion in turbulence, i.e. $X(t)$ {\em cannot} be directly written as a time-independent phase variable. Instead, it is given implicitly via evolution operators given in Table 1. For the stochastic turbulence model considered in this section, we also have the explicit expression using the stochastic, second-kind Volterra integral equation:
\begin{equation}
\begin{dcases}
X_1(t)&=\Tr e^{\int_0^t\L(s,\xi)ds}X_1(0)=+\int_0^t\partial_{X_2}\psi(X(s),s)ds+\kappa_1\D_1(t)\\
X_2(t)&=\Tr e^{\int_0^t\L(s,\xi)ds}X_2(0)=-\int_0^t\partial_{X_1}\psi(X(s),s)ds+\kappa_2\D_2(t),
\end{dcases}
\end{equation}
where $\Tr$ is the time-ordering operator and $\L(s,\xi)$ is the time-dependent {\em stochastic} generator \cite{zhu2020hypoellipticity} corresponding to the Brownian flow \eqref{eqn:alpha_beta}. The merit to consider the extended stochastic system \eqref{eqn:alpha_beta} is that $X=X(t)$ becomes a time-independent observable. Hence for fully developed turbulence where the statistical properties of the system can be characterized by the NESS probability density $\rho_S$, the information about the stationary turbulence dispersion is encoded in the dynamics of the reduced quantity $X(t)$ in $\rho_S$. Under the current setting, each trajectory $X(t)=X(t; x_0)$ of the Lagrangian particle can be viewed as a sample solution of the SDE \eqref{eqn:alpha_beta} evolving from random initial state $x(0)=[\{\alpha_l(0)\},\{\beta_l(0)\},X_1(0),X_2(0)]\sim\rho_S$.

We end this section with some discussion on the mathematical properties of the stochastic system \eqref{eqn:alpha_beta0} and \eqref{eqn:alpha_beta_simplified}. As a Markov process, the smoothness of the transition density $p_t(\{\alpha_l\},\{\beta_l\})$ for the 2D SNS equation \eqref{eqn:alpha_beta0} were proved by E and Mattingly \cite{weinan2001ergodicity}. Moreover, the stochastic flow induced by \eqref{eqn:alpha_beta0} is shown to be geometrically ergodic and approaches to the unique invariant measure $\rho_S\prod_{l}d\alpha_ld\beta_l$ as $t\rightarrow+\infty$. This means the limit $\lim_{t\rightarrow+\infty}p_t(\{\alpha_l\},\{\beta_l\})=\rho_S$ converges exponentially. The geometric ergodicity can also be written adjointly as an estimate for the state space observable $O(t)=O(\{\alpha_l(t)\},\{\beta_l(t)\})$:
\begin{align*}
    \|e^{t\K}O(0)-\langle O(0)\rangle_{\rho_S}\|\leq Ce^{-rt}
\end{align*}
with constants $C=C(O(0))>0,r>0$, where $\K$ is the corresponding Kolmogorov backward operator and $\langle\cdot\rangle_{\rho_S}$ is the ensemble average with respect to $\rho_S$. For the extended stochastic dynamics \eqref{eqn:alpha_beta_simplified}, 
immediately we can get the smoothness of the transition density $p_t(\{\alpha_l\},\{\beta_l\},X_1,X_2)$ for this diffusion process:
\begin{prop}\label{lemma1}
Let $K_1=\{(0,1),(1,1)\}$ and $K_2=\{(1,0),(1,1)\}$. If $K_1\subset K$ or $K_2\subset K$ and $\kappa_1,\kappa_2>0$, then the Markov process \eqref{eqn:alpha_beta} has a smooth transition probability density $p_t(\{\alpha_l\},\{\beta_l\},X_1,X_2)$ for $t>0$. 
\end{prop}
\begin{proof}
This is a corollary of the Lemma 1.2 in \cite{weinan2001ergodicity}, in which, the same result for the truncated SNS equation \eqref{eqn:NS_w_Truncated} was obtained by showing that the Lie algebra generated by the vector field $\{\X_0,\X_1,\cdots,\X_{2N}\}$ has full rank in every point of $\R^{(N+1)^2-1}$. For the extended system \eqref{eqn:alpha_beta_simplified}, the full rank properties of the Lie algebra generated by $\{\X_0,\X_1,\cdots,\X_{2N},\partial_{X_1},\partial_{X_2}\}$ in $\R^{(N+1)^2+1}$ is a direct result of the assumption $\kappa_1,\kappa_2>0$. 
\end{proof}
Following the argument used in \cite{weinan2001ergodicity}, it is not hard to obtain the existence and uniqueness of the invariant measure for SDE \eqref{eqn:alpha_beta}. The ergodicity is harder to prove for this extended system because the standard Lyapunov function method introduced in \cite{weinan2001ergodicity} cannot be directly applied. This problem was recently solved by Bedrossian et al. \cite{bedrossian2019almost} with a slightly different stochastic analysis framework.
\section{Generalized fluctuation-dissipation relations}\label{sec:FDT}
In nonequilibrium statistical mechanics, a standard technique to deal with the transport problem is the fluctuation-dissipation relation (FDR) \cite{kubo1966fluctuation}. Generally speaking, the classical first FDR gives the response result of an equilibrium system observable to the external perturbations, and the second FDR connects the memory kernel of the generalized Langevin equation (GLE) for such a equilibrium system observable with the time autocorrelation of the fluctuation force. The classical FDRs were first proved for canonical Hamiltonian systems \cite{kubo1966fluctuation} and have become the cornerstone of the nonequilibrium statistical mechanics which enable us to study various transport phenomena quantitatively. Over the years, the generalization of FDRs to open systems in the NESS has been studied by many researchers \cite{agarwal1972fluctuation,baiesi2009fluctuations,baiesi2009nonequilibrium,marconi2008fluctuation,majda2005information,gritsun2008climate,zhu2021hypoellipticity3}. Due to the limitation of our knowledge, here we only mentioned works which directly inspired the current paper and refer to review papers such as \cite{maes2020response,marconi2008fluctuation} and the reference therein for detailed explorations.   

As we mentioned before, viscous fluids are typical nonequilibrium systems. Hence we are lead to study the {\em generalized} FDRs which are valid in the nonequilibrium regime. Specifically, we will derive the operator-form first FDR and the path-integral-form first FDR for a general observable $O(x(t))$ in the nonequilibrium, and then use the resulting response formula for the Lagrangian particle $O(x(t))=X(t)$ to study the turbulent dispersion. In addition, we use the Mori-Zwanzig (MZ) formalism to prove a generalized second FDR for the generalized Langevin equation (GLE) of the Lagrangian particle. This leads to an effective reduced-order model which imitates the dynamics of $X(t)$. We emphasize {\em a priori} that the presented results are based on the 2D SNS equation, but the methodology can be readily generalized to other fluid systems. 

\subsection{The first FDR for the Lagrangian particle}
\label{sec:1st_FDR}
\subsubsection{Operator-form first FDR}
\label{sec:1st_FDR_operator}
To derive the generalized first FDR for the 2D SNS equation, we begin with a review of the general theory for SDEs. To this end, we consider the following $d$-dimensional, time-homogeneous stochastic dynamical system driven by Gaussian white noise:
\begin{align}\label{eqn:general_SDE}
dx(t)=F(x(t))dt+\sigma(x(t))d\W(t),\qquad x(0)\sim \rho,
\end{align}
where $d\W(t)$ is the standard Wiener process with $d\W_i(t)$ independent with each other and $\rho$ is the distribution of the initial condition. As far as we are concerned, the generalized FDR for stochastic systems such as \eqref{eqn:general_SDE} was first derived by Agarwal \cite{agarwal1972fluctuation} and also discovered by many different researchers \cite{falcioni1990correlation,baiesi2009nonequilibrium,dal2019linear,majda2005information}. We will use the following version of the first FDR  \cite{zhu2021hypoellipticity3} to derive the linear response formula for fluid systems.
\begin{theorem}\label{thm1:1st-FDT}(Generalized-1st-FDR) Consider a general perturbation of the stochastic system \eqref{eqn:general_SDE}:
\begin{align}\label{eqn:sde_p}
dx(t)=F(x(t))dt+\sigma(x(t))d\W(t)+\delta G(t)\cdot\tilde{F}(x(t))dt+\sqrt{\delta} H(t)\cdot\tilde{\sigma}(x(t))d\W(t), \qquad x(0)\sim \rho,
\end{align}
where $\delta\ll 1$. We further assume that $\rho$ is the steady state distribution of the SDE \eqref{eqn:general_SDE} and it is a smooth function of $x$ which decays to 0 as $\|x\|\rightarrow \infty$. Then the following generalized first FDR holds for state space observable $O(t)=O(x(t))$:
\begin{equation}\label{Generalized_1st_FDT}
\begin{aligned}
\la O(t)\rangle_{\rho_{\delta}}-\la O(t)\ra=-\sum_{i=1}^d&\int_0^{t}\left\la\frac{1}{\rho} [\partial_{x_i}\tilde{F}_i\rho]O(t-s)\right\ra\delta G_i(s)ds\\
&+\sum_{i=1}^d\int_0^t\left\la\left(\partial_{x_i}^2\tilde{\sigma}_i+\frac{2}{\rho}\partial_{x_i}\tilde{\sigma}_i\partial_{x_i}\rho+\frac{1}{\rho}\tilde{\sigma}_i\partial_{x_i}^2\rho\right)O(t-s)\right\ra\delta H_i(s)ds+O(\delta^2).
\end{aligned}
\end{equation}
where $\rho_{\delta}$ is the steady state distribution of the perturbed system \eqref{eqn:sde_p}.
\end{theorem}
Linear response formula \eqref{Generalized_1st_FDT} is called the {\em operator-form first FDR} since its derivation \cite{zhu2021hypoellipticity3,kubo1966fluctuation} only uses the perturbation expansion of the evolution operator $e^{t\K}$. We note that it requires to know the exact form of $\rho$ in order to apply this formula since functions of $\rho$, e.g. $[\partial_{x_i}\tilde F_i\rho]/\rho$, exist in the ensemble average $\langle\cdot\rangle_{\rho}$. Applying this result to SDE \eqref{eqn:alpha_beta_simplified} for observable $X(t)$, we can get the linear response relation for Lagrangian particle dispersion. Change symbolically the parameter  $\nu,\kappa_1,\kappa_2,\sigma_l,\gamma_l$ in \eqref{eqn:alpha_beta_simplified}  as $\nu',\kappa_1',\kappa_2',\sigma_l',\gamma_l'$, then we get: 
\begin{prop}\label{cor1:1st-FDT(ver1)}For a fixed quantity $\delta\ll 1$, if the parameters $\nu',\kappa_1',\kappa_2',\sigma_l',\gamma_l'$ satisfy $\sum_l^N\nu'|l|^2\sim O(\delta)$, $\kappa_1',\kappa_1'\sim O(\delta)$ and $\sum_{l\in K}^N(\sigma_l'+\gamma_l')\sim O(\delta)$, where $\sim$ means approximately as, then we have the following linear response relation for the Lagrangian particle $X(t)$:
\begin{equation}\label{eqn:Linear_response}
\begin{aligned}
\la X_i(t)&\rangle_{\rho_{\delta}}-\la X_i(t)\ra=2\sum_{l}^{N}\nu'|l|^2\int_0^{t}\la X_i(s)\ra ds+\sum_{l}^N\int_0^t\nu'|l^2|\left\la \left(\frac{\alpha_l}{\rho}\partial_{\alpha_l}\rho +\frac{\beta_l}{\rho}\partial_{\beta_l}\rho \right)X_i(s)\right\ra ds\\
&+\int_0^t\left\la\left(\frac{\kappa_1'}{\rho}\partial_{X_1}^2\rho+\frac{\kappa_2'}{\rho}\partial_{X_1}^2\rho\right)X_i(s)\right\ra ds+\sum_{l\in K}^N\int_0^t\left\la\left(\frac{\sigma_l'}{\rho}\partial_{\alpha_l}^2\rho+\frac{\gamma_l'}{\rho}\partial_{\beta_l}^2\rho\right)X_i(s)\right\ra ds+O(\delta^2).
\end{aligned}
\end{equation}
\end{prop}
\noindent
By specifying different choice of the parameter set $\nu',\kappa_1',\kappa_2',\sigma_l',\gamma_l'$, we can obtain the equilibrium and nonequilibrium linear response formula for $X(t)$. Specifically, we have: 
\paragraph{\textbf{Equilibrium linear response}} When $\{\nu',\kappa_1',\kappa_2',\sigma_l',\gamma_l'\}$=$\{\nu,\kappa_1,\kappa_2,\sigma_l,\gamma_l\}$, then the viscosity term, random noise term and the molecular diffusion term in the extended NS system \eqref{eqn:alpha_beta_simplified} are small quantities, therefore \eqref{eqn:alpha_beta_simplified} can be viewed as a perturbed (truncated) Euler equation. Following the turbulence thermalization argument, if $\rho$ is specified to be the equilibrium ensemble \eqref{inv:Euler_micro_T} or \eqref{inv:Euler_cano}, then \eqref{eqn:Linear_response} yields a {\em equilibrium linear response formula} for turbulence. It allows, at least formally, the calculation of the averaged turbulent dispersion $\langle X(t)\rangle_{\rho_{\delta}}$ for fluid systems near the equilibrium. 
The near equilibrium fluid system can be the highly turbulent deterministic NS flow with $0<\nu\ll 1,0\leq \kappa_1\ll 1,0\leq \kappa_2\ll 1,\sigma_l=\gamma_l=0$ and the highly turbulent SNS flow with $0<\nu\ll 1$, $0\leq\kappa_1\ll 1$, $0\leq\kappa_2\ll 1,0<\sigma_l\ll 1,0<\gamma_l\ll 1$. The main computational challenge is the determination of the {\em extended} system stationary probability density $\rho$, which should satisfy the steady state Fokker-Planck equation $\K^*\rho=0$, where $\K^*$ is the adjoint of operator \eqref{op:K_NS_extend} with $\sigma_l=\gamma_l=0$.
\paragraph{\textbf{Nonequilibrium linear response}} When $\{\nu',\kappa_1',\kappa_2',\sigma_l',\gamma_l'\}$=$\{\nu(\delta),\kappa_1(\delta),\kappa_2(\delta),\sigma_l(\delta),\gamma_l(\delta)\}$, where terms such as $\nu(\delta)$ means a value in the small (of the order $O(\delta)$) neighborhood of an arbitrary value $\nu$. Then \eqref{eqn:Linear_response} is a {\em nonequilibrium linear response formula} in the steady state $\rho$, where $\rho$ satisfies $\K^*\rho=0$ with parameter $\{\nu,\kappa_1,\kappa_2,\sigma_l,\gamma_l\}$. This formula can be applied to calculated the averaged turbulent dispersion $\langle X(t)\rangle_{\rho_{\delta}}$ when the fluid system slightly deviate from the NESS. The perturbation can comes from a change of the fluid viscosity, i.e. $v(\delta)\Delta \omega$, or the stochastic force, or the Lagrangian particle molecular diffusivity. Without knowing the NESS distribution $\rho_S$ explicitly, the evaluation of \eqref{eqn:Linear_response} is generally harder than the equilibrium case.    
\subsubsection{Path-integral-form first FDR}
\label{sec:1st_FDR_path}
The above operator-form first FDR can be proved rigorously using standard perturbation analysis \cite{agarwal1972fluctuation,kubo1966fluctuation,zhu2021hypoellipticity3}. However, one may find that it is hard to use because the steady state distribution $\rho$ appears in the response function, which is generally difficult to obtain due to the high-dimensionality of the steady state Fokker-Planck equation $\K^*\rho=0$. 
This problem is more serious when we deal with the turbulent dispersion problem since even for equilibrium systems, one needs to solve for the stationary probability density $\rho$ of the {\em extended} system \eqref{eqn:alpha_beta_simplified} in order to apply linear response formula \eqref{eqn:Linear_response}. In recent years, another type of first FDR was developed \cite{baiesi2009nonequilibrium,baiesi2010nonequilibrium,baiesi2009fluctuations} and provided us a general nonequilibrium response formula which does not rely on the prior knowledge of the steady state distribution $\rho$. The method uses path-integral representation of the functional probability density associated with SDEs, hence can be termed as the {\em path-integral-form first FDR}. 

In this subsection, we will briefly explain the physical meaning of this new formula and then derive heuristically the path-integral-form first FDR for turbulent systems using the aforementioned analogy between the fluid equations and the heat conduction models. We further show that the obtained response formula can be rigorously proved for equilibrium fluid system with canonical equilibrium measure \eqref{inv:Euler_cano} and its extension for nonequilibrium fluid systems is direct. Since many formulas will be given, some of them can be proved while others are just based on physical arguments. For readers' convenience, we state in advance the mathematical rigor of the obtained results. First, interpreting the functional Radon-Nikodym derivative \eqref{Radon-Nikodym_derivative} as the total entropy production is a physical argument, although it makes sense intuitively. Its exact form can be given explicitly {\em only} for stochastic systems with the {\em non-degenerate} white noise. Hence the equilibrium formula \eqref{R(t)_Omega_2p}-\eqref{Mobility_X(t)}, the nonequlibrium formula \eqref{Mobility_X(t)2}-\eqref{R(s)_non-degenerate} and the nonlinear response formula \eqref{nonlinear_response}-\eqref{R(s)_nonlinear_response} are solid results. For degenerate stochastic system, due to the similarity of the heat conduction model with the fluid equation and the fact that the response formula for the former is verified in \cite{maes2020response}, we think formulas such as \eqref{R(s)_fluid_general}, \eqref{Mobility_X(t)_1} and \eqref{R(s)_nonlinear_response1} are physically reasonable. The response formula \eqref{R(s)_fluid_NESS_deter} and \eqref{N(t)_deter} for deterministic turbulence are obtained using a somewhat dubious singular perturbation argument, readers may treat them as conjectures.  
To motivate the discussion, we first note an important observation by Gallavotti \cite{gallavotti1996chaotic} which is also emphasised by Baiesi et al.\cite{baiesi2009nonequilibrium}: The equilibrium system linear response formula is not just the outcome of a pertubative calculation of the evolution operator, but can also be understood as the correlation between the perturbative force and the entropy production induced by it. By analyzing some prototype examples such as the Markov jump process and the overdamped Langevin dynamics, Baiesi et al \cite{baiesi2009nonequilibrium,baiesi2010nonequilibrium,baiesi2009fluctuations,maes2020response} showed that for nonequilibrium systems, aside from the entropy flux contribution $\S(t)$, a nonequilibrium extension of the linear response formula should includes another term which encodes the influence to the observable due to the change of the dynamical activity. This contribution is called as the {\em frenecy}, denoted as $\mathcal{T}(t)$. These two terms constitute the {\em change of the total entropy production} for a system transient from one nonequilibrium state to another nonequilibrium state. In the end, one may express the path-integral probability density functional $P[x^{\delta}(t)]$ for perturbed system as 
\begin{align}\label{Radon-Nikodym_derivative}
    P[x^{\delta}(t)]=
    \exp\left\{\frac{1}{k_B}\int_0^t\delta\E(s)ds\right\}P[x(t)]=
    \exp\left\{\frac{1}{2}\S(t)-\frac{1}{2}\mathcal{T}(t)\right\}P[x(t)]
\end{align}
where $x(t)$ and $x^{\delta}(t)$ are the paths of the underlying perturbed and unperturbed stochastic system. $P[x(t)]$ is the path-integral probability density functional corresponding to the unperturbed dynamics. $\delta\E(s)$ is the change of the entropy production rate. Formula \eqref{Radon-Nikodym_derivative} can be understood as a functional Radon-Nikodym derivative. As a result, a generalized response formula can be obtained using the following path-integral ensemble average formula:
\begin{equation}\label{ensemble_av}
\begin{aligned}
\langle O(x(t))\rangle_{\rho_{\delta}}&=\int \D[x(t)]P[x^{\delta}(t)|x(0)]O(x(t))\\
&=\int\D[x(t)] \frac{P[x^{\delta}(t)|x(0)]}{P[x(t)|x(0)]}P[x(t)|x(0)]O(x(t))\\
&=\int\D[x(t)] \frac{P[x^{\delta}(t)]}{P[x(t)]}P[x(t)]O(x(t))\\
&=\int\D[x(t)] \exp\left\{\frac{1}{k_B}\int_0^t\delta\E(s)ds\right\}P[x(t)]O(x(t))=\left\langle e^{\frac{1}{k_B}\int_0^t\delta\E(s)ds}O(x(t))\right\rangle_{\rho},
\end{aligned}
\end{equation}
which is the pre-announced formula \eqref{intro_entropy_prod_FDR}. In \eqref{ensemble_av}, $\D[x(t)]$ is the functional differential and $x(0)\sim \rho$. Due to the independence of the randomness of $x(0)$ and the path $x(t)$, we used the fact that $P[x(t)|x(0)]=P[x(t)]$, $P[x^{\delta}(t)|x(0)]=P[x^{\delta}(t)]$. If otherwise stated, in this subsection we use the notation $\langle\cdot\rangle_{\mu}$ to represent the average of the form \eqref{ensemble_av}, which is with respect to the {\em dynamical} ensemble of the path $x(t)$ with initial condition $x(0)\sim\mu$. By invoking the perturbation expansion $P[x^{\delta}(t)]/P[x(t)]=1+\frac{1}{2}\S(t)-\frac{1}{2}\mathcal{T}(t)+O(\delta^2)$, where $\frac{1}{2}\S(t)-\frac{1}{2}\mathcal{T}(t)$ is a function of the path $x(t)$ of the magnitude $O(\delta)$, we get the generalized Kubo's linear response formula:
\begin{equation}\label{ensemble_av1}
\begin{aligned}
\langle O(x(t))\rangle_{\rho_{\delta}}&=\int\D[x(t)]P[x(t)|x(0)]O(x(t))+\int\D[x(t)]\left[\frac{1}{2}\S(t)-\frac{1}{2}\mathcal{T}(t)\right]P[x(t)|x(0)]O(x(t))+O(\delta^2)\\
&=\langle O(x(t))\rangle_{\rho}+\left\langle \left[\frac{1}{2}\S(t)-\frac{1}{2}\mathcal{T}(t)\right]O(x(t))\right\rangle_{\rho}+O(\delta^2).
\end{aligned}
\end{equation}

As an application of this methodology, we now review Maes' result (Example IV.2 in \cite{maes2020response}) on the path-integral-form first FDR for a heat conduction model. To this end, we consider the following perturbation to the heat conduction model with Langevin thermostats, i.e. Eqn \eqref{SDE:n_d_heat}:
\begin{equation}\label{SDE:heat_perturbation}
\begin{dcases}
\begin{aligned}
dq_i&=p_idt\\
dp_i&=-\partial_{q_i}H dt-\gamma_ip_idt+\sqrt{2k_BT_i\gamma_i}d\W_i(t)
\end{aligned}
\end{dcases}
\quad
\xRightarrow[]{-\delta\mu p_i}
\quad
\begin{dcases}
\begin{aligned}
dq_i&=p_idt\\
dp_i&=-\partial_{q_i}H dt-\gamma_i'p_idt+\sqrt{2k_BT_i\gamma_i}d\W_i(t),
\end{aligned}
\end{dcases}
\end{equation}
where $\gamma_i'=\gamma_i+\delta\mu$. Here we assume that each particle is in contact with the thermostat of temperature $T_i$, hence $\gamma_i>0$ for all $i\in\G$. For any observable function $O(t)=O(x(t))=O(p(t),q(t))$ (a typical example is the averaged heat flux \eqref{def:local_global_flux}), we have the linear response formula 
\begin{align}\label{heat_linear_response}
    \langle O(t)\rangle_{\rho_{\delta}}-\langle O(t)\rangle_{\rho}=\int_0^t R_{O,p}(t,s)ds=\delta\chi_{O,p}(t),
\end{align}
where $R_{O,p}(t,s)$ is known as the response function and the subscript $O$ and $p$ indicates it is the response of $O(t)$ with respect to the perturbation at the momentum coordinates, i.e. $-\delta\mu p_i$. Its time-integral dividing the perturbation strength $\delta$ yields $\chi_{O,p}(t)$, which is the generalized mobility of the observable $O(t)$ in the NESS. Using the path-integral method, we can represent $R_{O,p}(t,s)$ as 
\begin{align}\label{Mobility}
R_{O,p}(t,s)
=\left.\frac{\partial}{\partial\delta}\langle O(t)\rangle_{\rho_{\delta}}\right\vert_{\delta=0}
=\frac{1}{2}\left\langle\frac{\partial\S(t)}{\partial\delta}O(t)\right\rangle_{\rho}-\frac{1}{2}\left\langle\frac{\partial\mathcal{T}(t)}{\partial\delta}O(t)\right\rangle_{\rho}
=\langle \mathcal{R}(s)O(t)\rangle_{\rho}.
\end{align}
To avoid confusion, hereafter we used notation $\partial/\partial\delta$ to represent the functional derivative with respect to $\delta$. In \eqref{Mobility}, $\S(t)$ and $\mathcal{T}(t)$ are respectively the excess of entropy flux and frenesy contributions to the response quantity $\mathcal{R}(s)$. Specifically, we have
\begin{align}\label{S(t)_T(t)}
\S(t)=-\sum_{i=1}^{|\G|}\frac{\delta\mu}{k_BT_i}\int_0^tp_i^2(s)ds,
\quad
\mathcal{T}(t)=-\sum_{i=1}^{|\G|}
\frac{2\delta\mu}{k_BT_i\gamma_i}\int_0^tp_i^2(s)ds+\sum_i^{|\G|}\frac{2\sqrt{2}\delta\mu}{\sqrt{k_BT_i\gamma_i}}\int_0^tp_i(s)d\W_i(s)+O(\delta^2),
\end{align}
where $|\G|$ is the cardinality of set $\G$ which is the total number of particles of the heat conduction model. This indicates that the whole system has $|\G|$ thermostats. We note that $\mathcal{T}$ can be equivalently represented as:
\begin{align}
\mathcal{T}(t)=\tau(t)+O(\delta^2)=\sum_{i=1}^{|\G|}
\frac{\delta\mu}{k_BT_i\gamma_i}p_i^2(s)+\sum_i^{|\G|}\frac{2\delta\mu}{\sqrt{k_BT_i\gamma_i}}\int_0^tp_i(s)\partial_{q_i}H(s)ds+O(\delta^2).
\end{align}
The physical meaning of the excess entropy flux and frenecy in \eqref{S(t)_T(t)} can be interpreted as follows. For each $i\in\B$, $-\delta\mu\int_0^tp_i^2(s)ds/T_i$ is the heat transmitted to the reservoir divided by the temperature, therefore is the excess of entropy flux created by the $i$-th thermostat. $\S(t)$ is then the total excess of entropy flux for the heat conduction model. On the other hand,
\begin{align}
-\frac{2\delta\mu}{k_BT_i\gamma_i}p_i(s)ds+\frac{2\sqrt{2}\delta\mu}{\sqrt{k_BT_i\gamma_i}}d\W_i(s)+O(\delta^2)
\end{align}
is the infinitesimal change of the heat conduction model that has direct interactions with $p_j(s)$. Hence $\mathcal{T}(t)$ accounts the total excess of dynamical activity (frenecy) of the stochastic model. With $\S(t),\mathcal{T}(t)$ defined as \eqref{S(t)_T(t)}, we can get that the response quantity $\mathcal{R}(s)$ in \eqref{Mobility} has the form:
\begin{align}\label{R(s)}
\mathcal{R}(s)=\sum_{i=1}^{|\G|}\left(\frac{\delta\mu}{k_BT_i\gamma_i}-\frac{\delta\mu}{2k_BT_i}\right)p_i^2(s)-\sum_i^{|\G|}\frac{\sqrt{2}\delta\mu}{\sqrt{k_BT_i\gamma_i}}p_i(s)W_i(s),
\end{align}
where $d\W_i(s)=W_i(s)ds$. As we discussed in Section \ref{sec:Noneqn_viscous}, the 2D SNS equation \eqref{eqn:NS} with stochastic forces (the case with NESS2) can be interpreted as the heat conduction model with Langevin thermostats. This allows us to derive analogically the path-integral-form first FDR for the fluid system \eqref{eqn:alpha_beta0} and its extension \eqref{eqn:alpha_beta_simplified}. To this end, we consider the following perturbation of \eqref{eqn:alpha_beta_simplified}:
\begin{align}\label{FDR_perturbation}
    \begin{dcases}
    d\alpha_l&=F_{l}(\{\alpha_l\},\{\beta_l\})dt-\nu|l|^2\alpha_ldt+\sigma_ld\B_l\\
    d\beta_l&=G_{l}(\{\alpha_l\},\{\beta_l\})dt-\nu|l|^2\beta_ldt+\gamma_ld\W_l\\
    dX_1&=\partial_{X_2}\psi(X,t)dt+\kappa_1d\D_1(t)\\
    dX_2&=-\partial_{X_1}\psi(X,t)dt+\kappa_2d\D_2(t)
    \end{dcases}
    \quad
\xRightarrow[]{+\delta\mu\Delta \omega}
\quad
    \begin{dcases}
    d\alpha_l&=F_{l}(\{\alpha_l\},\{\beta_l\})dt-\nu'|l|^2\alpha_ldt+\sigma_ld\B_l\\
    d\beta_l&=G_{l}(\{\alpha_l\},\{\beta_l\})dt-\nu'|l|^2\beta_ldt+\gamma_ld\W_l\\
    dX_1&=\partial_{X_2}\psi(X,t)dt+\kappa_1d\D_1(t)\\
    dX_2&=-\partial_{X_1}\psi(X,t)dt+\kappa_2d\D_2(t),
    \end{dcases}
\end{align}
where $\nu'=\nu+\delta\mu$. To be noticed that in \eqref{FDR_perturbation}, the dynamics for the characteristic Lagrangian particle $X(t)$ is unchanged, hence in the linear order $O(\delta)$, the definition of $\S(t)$ and $\mathcal{T}(t)$ do not rely on $X(t)$. On the other hand, since the dynamics for $\alpha_l,\beta_l$ are independent of $X(t)$, for observable of the form $f(\{\alpha_l\},\{\beta_l\})$, any ensemble average with respect to the steady state distribution of the joint system \eqref{eqn:alpha_beta_simplified} equals to the ensemble average with respect to the steady state of Eqn \eqref{eqn:alpha_beta0}. In the following context, to distinguish it from the extended system, we use notation $\rho'$ to represent the average of the steady state of SNS equation \eqref{eqn:alpha_beta0}, with $\rho'=\rho_{eq}'$ corresponding to the equilibrium density and $\rho'=\rho_{S}'$ corresponding to the NESS density. Using this notation and imitating the definition of the excess of entropy flux for the heat conduction model, we may define $\S(t)$ for the fluid system as:
\begin{align}\label{entropy_flux}
    \S(t)=-\sum_l\frac{\delta\mu|l|^2}{\langle\alpha_l^2\rangle_{\rho'}}\int_0^t\alpha_l^2(s)ds-\frac{\delta\mu|l|^2}{\langle\beta^2_l\rangle_{\rho'}}\int_0^t\beta_l^2(s)ds.
\end{align}
We will explain later why the thermostat temperature $T_i$ is replaced by $\langle\alpha_l^2\rangle_{\rho'}$ and $\langle\beta_l^2\rangle_{\rho'}$. Similarly, the path-dependent frenesy can be defined as:
\begin{equation}\label{frenesy}
\begin{aligned}
\mathcal{T}(t)=\tau(t)+O(\delta^2)
=-\sum_l\frac{2\delta\mu\nu|l|^4}{\sigma_l^2}\int_0^t&\alpha_l^2(s)ds
+\frac{2\delta\mu|l|^2}{\sigma_l}\int_0^t\alpha_l(s)d\B_l(s)\\
&-\sum_l\frac{2\delta\mu\nu|l|^4}{\gamma_l^2}\int_0^t\beta_l^2(s)ds
+\frac{2\delta\mu|l|^2}{\gamma_l}\int_0^t\beta_l(s)d\W_l(s)+O(\delta^2).
\end{aligned}
\end{equation}
Naturally, we get the following expression for the response quantity $\mathcal{R}(s)$ in \eqref{Mobility}:
\begin{equation}\label{R(s)_fluid_general}
\begin{aligned}
\mathcal{R}(s)=\sum_{l}\bigg(\frac{\delta\mu\nu|l|^4}{\sigma_l^2}-\frac{\delta\mu|l|^2}{2\langle\alpha_l^2\rangle_{\rho'}}\bigg)\alpha_l^2(s)
+
\left(\frac{\delta\mu\nu|l|^4}{\gamma_l^2}-\frac{\delta\mu|l|^2}{2\langle\beta_l^2\rangle_{\rho'}}\right)\beta_l^2(s)
-\frac{\delta\mu|l|^2}{\sigma_l}\alpha_l(s)B_l(s)
-\frac{\delta\mu|l|^2}{\gamma_l}\beta_l(s)W_l(s),
\end{aligned}
\end{equation}
where $d\B_l(s)=B_l(s)ds$ and $d\W_l(s)=W_l(s)ds$. The above discussion is for the 2D SNS equation with NESS2. Based on these results and a singular perturbation analysis, we can derive a {\em tentative} path-integral-form linear response formula for the deterministic turbulence corresponding to NESS1 and NESS3. The trick is to view the deterministic NS equation as a singularly perturbed 2D SNS (in the sense of \cite{matkowsky1977exit}) in the limit of $\sigma_l,\gamma_l\rightarrow 0$. Relevant analysis in this regard is provided in \ref{app:determ_turbulence_FDR}. The final result we obtained is that the deterministic turbulence case (NESS1 and NESS3) has the following response quantity: 
\begin{equation}\label{R(s)_fluid_NESS_deter}
\begin{aligned}
\mathcal{R}(s)&=-\sum_{l}\frac{\delta\mu|l|^2}{2\langle\alpha_l^2\rangle_{\rho'}}\alpha_l^2(s)+\frac{\delta\mu|l|^2}{2\langle\beta_l^2\rangle_{\rho'}}\beta_l^2(s),
\qquad \text{(NESS1), (NESS3)}.
\end{aligned}
\end{equation}
We may interpret the physical meaning of this result as follows: In the discussion of the Gaussian thermostat model, we mentioned that when $\alpha$ equals a constant, then it can be interpreted as a ``bare'' heat conduction model. Without the Langevin thermostat, the system undergoes an thermal dynamical process where the system internal energy (Hamiltonian) dissipates and the local effective temperature is defined as $\langle\alpha_l^2\rangle_{\rho'},\langle\beta_l^2\rangle_{\rho'}$. During this process, the total entropy production rate can be calculated as the summation of the kinetic energy dissipation rate dividing the effective temperature. Thus the {\em change} of the entropy production rate is analogously defined as \eqref{R(s)_fluid_NESS_deter}. The final form we obtained happens to agree with the Gallavotti-Cohen \cite{gallavotti1995dynamical} fluctuation theorem if the relationship between the fluctuation theorem and the FDR is as we discussed in Section \ref{sec:sum_brief}. Although it makes sense intuitively, we have to admit \eqref{R(s)_fluid_NESS_deter} is obtained in a very dubious way since in the singular limit $\sigma_l,\gamma_l\rightarrow0$, the probability measures for the deterministic turbulence are not well-defined mathematical objects (see Table \ref{Tab:table}), and the dynamical features of the system are also different. However, we still retain these results here amid critiques even from ourselves, and wish the validity of \eqref{R(s)_fluid_NESS_deter} can be justified or disapproved via numerical simulations. For numerical studies, we note that if the well-developed deterministic turbulence is assumed to be ergodic, then $\langle\alpha_l^2\rangle_{\rho'}, \langle\beta_l^2\rangle_{\rho'}$ can be obtained using the time average.
\paragraph{\textbf{Equilibrium linear response}} We first consider the equilibrium system. As we discussion in Section \ref{sec:Equi_turbulence}, the canonical equilibrium measure for the 2D Euler equation \eqref{eqn:Euler} is given by \eqref{inv:Euler_cano}. If we rewrite the Hamiltonian $H$ as $\Omega_0$, the Casimir invariants $G_{2p}[\omega]$ as $\Omega_{2p}$, and the ``generalized" thermodynamic coefficients $\{\beta_0, \beta_p\}$, $p\geq 1$ as $\{\nu_p\}$, $p\geq 0$, then in the Fourier mode space, the canonical equilibrium measure with only one invariant can be rewritten as:
\footnote{We note that the canonical equilibrium measure \eqref{inv:Euler_cano} needs to be regularized only when $\Omega_0$, i.e. the Hamiltonian $H$, is included in the exponential part. Hence the summation $\sum_l$ in \eqref{inv_measure_enstrophy} can be interpreted as follows: Whenever $\Omega_0$ is included, $\sum_l=\sum_{l}^N$, where $N$ is maximum Fourier modes of the ultraviolet truncation. Otherwise, $\sum_l$ can also be $\sum_{l}^{\infty}$ and there is no anomaly. If we are mainly concerned with the regularized case and study from a numerical point of view, then one can simply interpret all summations in this subsection as $\sum_l=\sum_{l}^N$.}
\begin{align}\label{inv_measure_enstrophy}
\rho_{eq,p}'\prod_ld\alpha_ld\beta_l\propto\exp\left\{-\sum_l\nu_p|l|^{2p-2}(\alpha_l^2+\beta_l^2)\right\}\prod_ld\alpha_ld\beta_l=\exp\left\{-\nu_p\Omega_{2p}\right\}\prod_ld\alpha_ld\beta_l,
\end{align}
where the subscript $p$ indicates the conserved quantity of the measure. For example, if $\nu_p=\nu_0$, \eqref{inv_measure_enstrophy} gives to the energy-canonical ensemble $\rho_{eq,0}'\propto\exp\{-\nu_{0}\Omega_0\}$. If $\nu_p=\nu_1$, then we get the enstrophy-canonical ensemble $\rho_{eq_,1}'\propto\exp\{-\nu_0\Omega_2\}$. A key fact about the 2D SNS equation
\eqref{eqn:alpha_beta_simplified} is that if $\sigma_l^2=\gamma_l^2=|l|^{4-2p}$ and $\nu=\nu_p$, then \eqref{inv_measure_enstrophy} is also the invariant measure for the stochastic flow corresponding to \eqref{eqn:alpha_beta0}. The proof of this is provided in \ref{app:Enstrophy_measure_proof}. For such a case, one may simply interpret the relationship between $\nu$ and $\sigma_l,\gamma_l$ as the ``fluctuation-dissipation relation'' for 2D SNS equation. Hence the equation \eqref{eqn:alpha_beta0} itself can be understood as an ``overdamped Langevin dynamics'' for equilibrium system with invariant measure
\eqref{inv_measure_enstrophy}. Naturally \eqref{heat_linear_response} yields a linear response relation for the {\em equilibrium system}. To determine the specific form of the response function $R_{O,\Delta\omega}(t)$, we note that $\alpha_l$ and $\beta_l$ are i.i.d Gaussian variables according to \eqref{inv_measure_enstrophy}. Naturally we have $\langle\alpha_l^2\rangle_{\rho_{eq,p}'}=\langle\beta_l^2\rangle_{\rho_{eq,p}'}=1/(2\nu_p|l|^{2p-2})$. With these relations, the response quantity \eqref{R(s)_fluid_general} can be simplified as:
\begin{equation}\label{R(t)_Omega_2p}
\begin{aligned}
\mathcal{R}(s)=\mathcal{R}_{2p}(s)
=\frac{1}{2}\frac{\partial}{\partial\delta}\S(t)-\frac{1}{2}\frac{\partial}{\partial\delta}\tau(t)
=-\delta\mu\sum_l|l|^p(\alpha_l(s)B_l(s)+\beta_l(s)W_l(s)),\qquad p\geq 0.
\end{aligned}
\end{equation}
As a specific example, the enstrophy-canonical measure with $p=1$ is given by:
\begin{equation}\label{eqn_R(t)_enstrophy}
\begin{aligned}
\mathcal{R}(s)
=\mathcal{R}_2(s)
=-\delta\mu\sum_l|l|(\alpha_l(s)B_l(s)+\beta_l(s)W_l(s)).
\end{aligned}
\end{equation}

Now we prove the validity of the linear response formula \eqref{heat_linear_response} with response function \eqref{R(t)_Omega_2p} for the equilibrium system. Since $\sigma_l^2=\gamma_l^2=|l|^{4-2p}>0$, $\kappa_1,\kappa_2>0$, we note that the SDE \eqref{eqn:alpha_beta_simplified} has {\em non-degenerate} white noise.
Then we can use the functional Radon-Nikodym derivative \eqref{Radon-Nikodym_derivative} and the exact path-integral form probability density of a non-degenerate SDE (Stratonovich's formula (2.4.42) in \cite{moss1989noise}) to get the proof of \eqref{R(t)_Omega_2p}. The technical details are provided in \ref{app:Radon-Nikodym_proof}. Choosing $O(t)=X_i(t)$ in \eqref{Mobility}, we get the following {\em equilibrium linear response} formula for the Lagrangian particle $X(t)$, namely 
\begin{align}\label{Mobility_X(t)}
   R_{X,\Delta\omega}(t,s)
   =\langle \mathcal{R}(s)X_i(t)\rangle_{\rho_{eq}},
\end{align}
where $\mathcal{R}(s)$ is given by \eqref{R(t)_Omega_2p}. We note that in Eqn \eqref{Mobility_X(t)}, $\rho_{eq}$ is of the invariant measure of the extended system \eqref{eqn:alpha_beta_simplified}, instead of \eqref{inv_measure_enstrophy}.

%
%
%
\paragraph{\textbf{Nonequilibrium linear response}} When we do not have relations such as $\sigma_l^2=\gamma_l^2=|l|^{4-2p}$, the 2D SNS equation \eqref{eqn:alpha_beta0} yields a nonequilibrium dynamics where the NESS distribution $\rho_S$ is generally unknown, so does the extended SNS equation \eqref{eqn:alpha_beta_simplified}. The nonequilibrium can be further divided into two cases, where one has the {\em degenerate} noise and the other one is imposed with {\em non-degenerate} noise. When we say the noise is degenerate, we mean only a subset $K$ of all Fourier modes has non-zero diffusion coefficient $\sigma_l$ or $\gamma_l$. The proof in \ref{app:Radon-Nikodym_proof} already show that the response formula for the non-degenerate SNS equation is as:
\begin{align}\label{Mobility_X(t)2}
   R_{X,\Delta\omega}(t,s)
   =\langle \mathcal{R}(s)X_i(t)\rangle_{\rho_{S}},
\end{align}
where $\mathcal{R}(s)$ is proven to be 
\begin{equation}\label{R(s)_non-degenerate}
\begin{aligned}
\mathcal{R}(s)=\sum_{l}
-\frac{\delta\mu|l|^2}{\sigma_l}\alpha_l(s)B_l(s)
-\frac{\delta\mu|l|^2}{\gamma_l}\beta_l(s)W_l(s).
\end{aligned}
\end{equation}
On the other hand, for the degenerate case, a mathematically rigorous derivation of the nonequilibrium linear response formula is still lacking. In particular, it is hard to get an explicit expression of the probability density functional $P[x(t)]$. For such systems, since physically we can still define the excess entropy flux $\S(t)$ and frenecy $\mathcal{T}(t)$ as \eqref{entropy_flux} and \eqref{frenesy} respectively, one would expect the linear response formula with response function \eqref{R(s)_fluid_general} is valid. As an example, choosing $O(t)=X_i(t)$ in \eqref{Mobility}, we can get the following {\em nonequilibrium linear response} formula for Lagrangian particle $X(t)$, namely 
\begin{align}\label{Mobility_X(t)_1}
   R_{X,\Delta\omega}(t,s)
   =\langle \mathcal{R}(s)X_i(t)\rangle_{\rho_{S}},
\end{align}
where $\mathcal{R}(s)$ is given by \eqref{R(s)_fluid_general} and the ensemble average therein, i.e. $\langle\cdot\rangle_{\rho'}$, satisfies $\langle\cdot\rangle_{\rho'}=\langle\cdot\rangle_{\rho_S'}$, where $\rho_{S}'$ is the NESS distribution of \eqref{eqn:alpha_beta0}. However we note that in Eqn \eqref{Mobility_X(t)_1}, $\rho_{S}$ is of the invariant measure of the extended system \eqref{eqn:alpha_beta_simplified}. 

Lastly, we explain why for the thermostat temperature $T_i$ in \eqref{S(t)_T(t)} has to be replaced by $\langle\alpha_l^2\rangle_{\rho}$ and $\langle\beta_l^2\rangle_{\rho}$. For the heat conduction model \eqref{SDE:heat_perturbation}, {\em each} particle is connected to a thermostat with temperature $T_i$. Hence in the stationary regime, the kinetic temperature for the $i$-th particle satisfies $\langle p_i^2\rangle_{\rho}=T_i$. In reality, however, it is more common that only the boundary particles are connected with the thermostats, as we have shown in Figure \ref{fig:Heat_transport}. For such cases, the effective temperature for the $i$-th particle will no longer be $T_i$. In fact, as many numerical studies suggest \cite{lepri2003thermal}, the boundary resistance effect exists which leads to a difference between the thermostat's temperature and the contacting boundary particle's effective temperature. Due to the similarity between the nonequilibrium heat conduction model and the fluid equation, we argue that a physically meaningful replacement of the thermostat temperature $T_i$ in \eqref{S(t)_T(t)} would be the averaged kinetic energy $\langle\alpha_l^2\rangle_{\rho}$, $\langle\beta_l^2\rangle_{\rho}$ for the Fourier modes. This is by imitating the microscopic definition of the temperature for molecular heat conduction models \cite{lepri2003thermal}.

\paragraph{\textbf{Nonlinear response}} Another advantage of the path-integral-form first FDR is that it is easy to get the nonlinear response formula corresponding to large external stimuli. For equilibrium or nonequilibrium cases with the non-degenerate noise, using the ensemble average formula \eqref{ensemble_av} and the exact expression of the Radon-Nikodym derivative \eqref{RN-fluid}, we obtain a nonlinear response formula
\begin{align}\label{nonlinear_response}
\langle O(t)\rangle_{\rho_{\delta}}=\langle \mathcal{N}(t)O(t)\rangle_{\rho},
\end{align}
where $\mathcal{N}(t)$ is the change of the total entropy production:
\begin{equation}\label{R(s)_nonlinear_response}
\begin{aligned}
\mathcal{N}(t)= \exp\bigg\{-\int_0^t\sum_l\bigg(\frac{\delta\mu|l|^2}{\sigma_l}&\alpha_l(s)d\B_l(s)
+\frac{\delta\mu|l|^2}{\gamma_l}\beta_l(s)d\W_l(s)\bigg)\\
&-\int_0^tds\sum_l\bigg(\frac{\delta^2\mu^2|l|^4}{2\sigma_l^2}\alpha^2_l(s)
+\frac{\delta^2\mu^2|l|^4}{2\gamma_l^2}\beta^2_l(s)\bigg)\bigg\}.
\end{aligned}
\end{equation}
Its extension to the colored noise case follows directly the argument in \ref{app:Radon-Nikodym_proof}. If the noise is degenerate, we expect $\mathcal{N}(t)$ becomes:
\begin{equation}\label{R(s)_nonlinear_response1}
\begin{aligned}
\mathcal{N}(t)
=\exp\bigg\{\int_0^tds\sum_{l}&\bigg(\frac{\delta\mu\nu|l|^4}{\sigma_l^2}
-\frac{\delta\mu|l|^2}{2\langle\alpha_l^2\rangle_{\rho'}}
-\frac{\delta^2\mu^2|l|^4}{2\sigma_l^2}\bigg)\alpha_l^2(s)
+
\left(\frac{\delta\mu\nu|l|^4}{\gamma_l^2}
-\frac{\delta\mu|l|^2}{2\langle\beta_l^2\rangle_{\rho'}}
-\frac{\delta^2\mu^2|l|^4}{2\gamma_l^2}\right)\beta_l^2(s)\\
&-\int_0^t\sum_l\bigg(\frac{\delta\mu|l|^2}{\sigma_l}\alpha_l(s)d\B_l(s)
+\frac{\delta\mu|l|^2}{\gamma_l}\beta_l(s)d\W_l(s)\bigg)\bigg\}.
\end{aligned}
\end{equation}

\paragraph{\textbf{Remark}} We emphasize that all response formulas derived in this subsection are {\em directly computable} using the Monte-Carlo simulation. Namely, with the explicitly given response function $\mathcal{R}(s)$ or $\N(t)$, by sampling over the NESS $\rho$ and the solving numerically the SDE \eqref{eqn:alpha_beta_simplified}, one can get the ensemble average $\langle \mathcal{R}(s)X_i(t)\rangle_{\rho}$, therefore $\langle X(t)\rangle_{\rho_{\delta}}$. The ergodicity of the stochastic systems \eqref{eqn:alpha_beta0} and \eqref{eqn:alpha_beta_simplified} guarantees the unique NESS state can be achieved using a long-time simulation of \eqref{eqn:alpha_beta0} and \eqref{eqn:alpha_beta_simplified}, and the ensemble average can be replaced by the time average. The calculation formula we are lead to is the previously announced Eqn \eqref{intro_ensemble_average}. On the other hand, using the proof we have shown in \ref{app:Radon-Nikodym_proof}, it is straightforward to obtain  response formulas corresponding to other types of perturbations. Specifically, we note that the path-integral technique works for non-stationary perturbations such as $\delta(t)\Delta \omega$. It also applies to stochastic fluid systems such as the stochastic KdV equation and the 3D SNS equation.

\subsection{Mori-Zwanzig equation and the second FDR for the Lagrangian particle}\label{sec:2nd_FDR}
The first FDRs discussed in the previous section only provide the moment information about Lagrangian particle dispersion. To get an {\em effective} reduced-order model for $X(t)$ in the stationary regimes, one can understand $X(t)$ as a 2D stationary stochastic processes (normally non-Gaussian) and invoke the Kramers-Moyal expansion \cite{moss1989noise,Risken} for the evolution equation of the marginal probability density of $X(t)$:
\begin{align}\label{Kramers-Moyal}
    \partial_tp(X,t)=\sum_{n=1}^{\infty}\sum_{m=0}^n\frac{(-1)^n}{n!}\frac{\partial^n}{\partial^mX_1\partial^{(n-m)}X_2}[C_{X_1^m,X_2^{(n-m)}}(X)p(X,t)]
\end{align}
where $C_{X_1^m,X_2^{(n-m)}}(X)$ is the cumulant of the stochastic process $X(t)$ whose specific from is given by Stratonovich in \cite{moss1989noise}. The response formula obtained in the previous section can help to calculate the moment $C_{X_1^m,X_2^{(n-m)}}(X)$. Eqn \eqref{Kramers-Moyal} can be viewed as the Heisenberg picture description of the dynamics for $X(t)$. An equivalent Sch\"ordinger-picture description in the state space studies directly the evolution of $X(t)$. In particular, we can use the Mori-Zwanzig(MZ) formulation for stochastic systems \cite{zhu2021hypoellipticity3,espanol1995hydrodynamics,hudson2020coarse} to get to an {\em exact} reduced-order evolution equation for $X(t)$, which is known as the generalized Langevin equation (GLE). By approximating the resulting GLE with some proper method, one can get a stochastic process which directly models the movement of the Lagrangian particle \cite{zhu2021hypoellipticity3}. 

In this subsection, we will derive the MZ equation and prove a generalized second FDR for observable $X(t)$. Most of the presented materials are from our previous work \cite{zhu2021hypoellipticity3}. For the completeness of the article, we will briefly review the result obtained in \cite{zhu2021hypoellipticity3} and focus on its application to fluid systems. Interested readers may refer to \cite{zhu2021hypoellipticity3} for more technical details.   
To derive the MZ equation, we consider a general stochastic system \eqref{eqn:general_SDE}. Now we introduce a projection operator $\P$ and its orthogonal $\Q=\I-\P$ whose specific form of will be defined later. By differentiate the well-known Dyson's identity:
\begin{align*}
    e^{t\K}=e^{t\Q\K}+\int_0^te^{s\K}\P\K e^{(t-s)\Q\K}\Q\K ds,
\end{align*}
and applying it any reduced-order observable $O(t)$, we get the following MZ equation for the evolution of the noise-averaged observable $o(t)=\mathbb{E}_{\W}[O(t)|x(0)]$:
\begin{equation}
\frac{\partial}{\partial t}e^{t\K} o(0)
=e^{t\K}\mathcal{PK} o(0)
+e^{t\Q\K}\mathcal{QK} o(0)+\int_0^te^{s\K}\P\K
e^{(t-s)\Q\K}\mathcal{QK} o(0)ds,\label{eqn:EMZ_full}
\end{equation}
where $\K$ is the Kolmogorov backward operator corresponding to \eqref{eqn:general_SDE}. The three terms at the right hand side of \eqref{eqn:EMZ_full} 
are called, respectively, streaming term, fluctuation 
(or noise) term, and memory term.  Applying the projection operator 
$\P$ to \eqref{eqn:EMZ_full}
yields the projected MZ equation
\begin{equation} 
\frac{\partial}{\partial t}\mathcal{P}e^{t\K}o(0)
=\mathcal{P}e^{t\K}\mathcal{PK} o(0)
+\int_0^t\P e^{s\K}\mathcal{PK}
e^{(t-s)\Q\K}\mathcal{QK} o(0)ds.\label{eqn:EMZ_projected}
\end{equation}
Note that the MZ equation \eqref{eqn:EMZ_full} and its projected form \eqref{eqn:EMZ_projected} for stochastic systems 
has the same structure as the classical projected MZ 
equation for deterministic (autonomous) systems 
\cite{zhu2019generalized,zhu2018estimation,zhu2018faber}. 
However, the Liouville operator $\L$ is replaced by 
a Kolmogorov operator $\K$. To specify the projection operator $\P$,
we consider the weighted Hilbert space 
$H=L^2(\R^n,\rho)$ with inner product
\begin{equation}
\langle h,g\rangle_{\rho}=\int_{\R^n} 
h(x)g(x)\rho(x)dx 
\qquad h,g\in H,
\label{ip}
\end{equation}
where $\rho$ is a positive weight 
function in $\R^n$ which is often chosen as the
stationary probability density of the stochastic system. 
The Mori-type projection operator $\P$ is a finite-rank projection operator in $H$, which has the general form:
\begin{align}
\label{Mori_P}
\P h=\sum_{i,j=1}^N G^{-1}_{ij}
\langle o_i(0),h\rangle_{\rho}o_j(0),
\qquad h\in H.
\end{align}
where $G_{ij}=\langle o_i(0),o_j(0)\rangle_{\rho}$
and $o_i(0)=o_i(x(0))$ ($i=1,...,N$) are
$N$ linearly independent functions. It is easy to check that $\P$ and its orthogonal $\Q=\I-\P$ are both symmetric projection operators in $H$, i.e. $\P^*=\P=\P^2$, $\Q^*=\Q=\Q^2$. With $\P$ 
defined as in \eqref{Mori_P}, the memory integral of 
the MZ equation \eqref{eqn:EMZ_full} and its 
projected version \eqref{eqn:EMZ_projected} can be simplified to a convolution term. Therefore these two MZ equations can be rewritten as:
\begin{align}
\frac{d o(t)}{dt} &= \Omega o(t) +
\int_{0}^{t}M(t-s) o(s)ds+f(t),\label{gle_full}\\
\frac{d}{dt}\P{o}(t) &= \Omega\P {o}(t)+ 
\int_{0}^{t}M(t-s) \P {o}(s)ds,\label{gle_projected}
\end{align}
where $o(t) =[o_1(t),\dots,o_N(t)]^T$ and 
\begin{subequations}
\begin{align}
		G_{ij} & = \langle o_{i}(0), o_{j}(0)\rangle_{\rho}
		\quad \text{(Gram matrix)},\label{gram}\\
		\Omega_{ij} &= \sum_{k=1}^NG^{-1}_{jk}
		\langle o_{k}(0), \K o_{i}(0)\rangle_{{\rho}}\quad 
		\text{(streaming matrix)},\label{streaming}\\
		M_{ij}(t) & =\sum_{k=1}^N G^{-1}_{jk}
		\langle o_{k}(0), \K e^{t\Q\K}\Q\K o_{i}(0)\rangle_{\rho}\quad 
		\text{(memory kernel)},\label{memory_kernel}\\
		 f_i(t)& =e^{t\Q\K}\Q\K o_i(0) \quad 
		\text{(fluctuation term)}.\label{f}
	\end{align}
\end{subequations}
\noindent
Equation \eqref{gle_full}-\eqref{gle_projected} are known as the (linear) generalized Langevin equation (GLE) in statistical physics, which give the {\em exact} evolution equation for the reduced-order quantity $o(t)$ and $\P o(t)$. As it is shown in \cite{zhu2021hypoellipticity3}, when we choose $\rho(0)$ in \eqref{eqn:general_SDE} to be $\rho_{eq}$ for the equilibrium system or $\rho(0)$ to be $\rho_{S}$ for the nonequilibrium system, then the reduced-order quantity $o(t)$ yields the white-noise-averaged random path for $O(t)$ hence imitates the dynamics $O(t)$. Additionally, $\P o(t)=C(t)=\langle o(t),o(0)\rangle_{\rho}=\langle O(t),O(0)\rangle_{\rho}$ gives the stationary time autocorrelation function of $O(t)$ for $\rho=\rho_{eq}$ or $\rho=\rho_S$.

The projection operator method provides a systematic way to derive the reduced-order model for observable $O(t)$ from the first principle. For deterministic Hamiltonian systems, the merit of using the MZ framework to derive the GLE is that the classical second fluctuation-dissipation relation (FDR) holds naturally as the result of the skew-symmetry of the Liouville operator $\L$ \cite{snook2006langevin,zhu2018estimation}. For stochastic systems studied in this paper, since the Kolmogorov operator $\K$ is not skew-symmetric, the classical second FDT has to be generalized accordingly. Applying Theorem 2 proved in \cite{zhu2021hypoellipticity3} to the SNS equation, we obtain the following result:
\begin{theorem}\label{thm:2nd-FDT}
(Generalized-2nd-FDR) For extended SNS system \eqref{eqn:alpha_beta_simplified}, if we assume that the steady state distribution $\rho$ decays to $0$ as $\|\alpha_l\|,\|\beta_l\|,\|X_1\|,\|X_2\|\rightarrow\infty$, then the following generalized second FDR holds for any state space observable $o(t)=o(x(t))$ in the GLE \eqref{gle_full}-\eqref{gle_projected}:
\begin{align}\label{general_2nd_FDT}
M_{ij}(t)=\sum_{k=1}^N G^{-1}_{jk}(-\langle f_{k}(0), f_i(t)\rangle_{\rho}+\langle \A o_{k}(0), f_i(t)\rangle_{\rho}),\qquad 1\leq i,j\leq N.
\end{align}
where $\A$ is a second-order differential operator defined as 
\begin{align}\label{W_exact_form}
\A=\sum_{j\geq 1}\X_j^2+\X_j(\ln\rho)\X_j,
\end{align}
and $\X_j$ is the vector field defined in \eqref{eqn:X_k*X_k}.
\end{theorem}
The above result holds regardless the value of parameters $\nu,\sigma_l,\gamma_l$ in 
\eqref{eqn:alpha_beta_simplified}, therefore is valid for both equilibrium systems and nonequilibrium systems. In the following, we will use the unified notation $\rho$ to represent the equilibrium distribution $\rho=\rho_{eq}$ and the nonequilibrium stationary distribution $\rho=\rho_S$. If we choose observable $O(t)$ in \eqref{Mori_P} to be the Lagrangian particle position $X(t)$, then 
\begin{align}\label{Mori_P1}
    \P h=\sum_{i,j=1}^2G_{i,j}^{-1}\langle X_i(0),h\rangle_{\rho}X_j(0),\qquad h\in L^2(\R^{(N+1)^2+1},\rho),
\end{align}
and the Mori-type GLEs for $X(t)$ are given by:
\begin{align}
\frac{d}{dt} x(t)&= \Omega x(t) +
\int_{0}^{t}M(t-s) x(s)ds+f_x(t),\label{gle_full_X}\\
\frac{d}{dt}C(t) &= \Omega C(t)+ 
\int_{0}^{t}M(t-s) C(s)ds,\label{gle_projected_X}
\end{align}
where $x(t)=\mathbb{E}_{\B_l,\W_l,\D_1,\D_2}[X(t)|\{\alpha_l(0)\},\{\beta_l(0)\},X_1(0),X_2(0)]$ with $\{\alpha_l(0)\},\{\beta_l(0)\},X_1(0),X_2(0)\sim \rho$ and $C(t)=\langle X_i(t),X_j(0)\rangle_{\rho}$ is the steady state time autocorrelation function of the Lagrangian particle $X(t)$. According to Theorem \ref{thm:2nd-FDT}, the following generalized second FDR holds for the GLE \eqref{gle_full_X}-\eqref{gle_projected_X}:
\begin{prop}\label{cor:2nd-FDT}
With the Mori-type projection operator \eqref{Mori_P1}, the following generalized second FDR holds for $X(t)$ in GLEs \eqref{gle_full_X}-\eqref{gle_projected_X}:
\begin{align}\label{general_2nd_FDT1}
M_{ij}(t)=\sum_{k=1}^2 G^{-1}_{jk}\left(-\langle f_{k}(0), f_i(t)\rangle_{\rho}+\langle \kappa_{k}^2\partial_{X_k}(
\ln\rho), f_i(t)\rangle_{\rho}\right),\qquad 1\leq i,j\leq 2.
\end{align}
\end{prop}
\begin{proof}Since $X_i$ is not a function of $\alpha_l,\beta_l$, we have $\X_jX_i(0)=0$ for all the vector field $\X_j$ with respect to $\alpha_l,\beta_l$. Then \eqref{general_2nd_FDT1} can be obtained by simple calculations.
\end{proof}
It is worth noticing that is the Lagrangian particle is purely passive, i.e. the molecular diffusion $\kappa_{1}=\kappa_2=0$, then the {\em classical} second FDR is valid for $X(t)$. The GLE \eqref{gle_full_X}-\eqref{gle_projected_X} provide us the path-space information about the stationary turbulent dispersion. Since the generalized second FDR \eqref{general_2nd_FDT1} links the the memory kernel $K(t)$ with the the fluctuation force $f_x(t)$, to get the specific form of the effective reduced-order model \eqref{gle_full_X}, it boils down to the approximation of the MZ memory kernel $K(t)$. This is a problem which is worth independent investigations. Here we only note some recent results \cite{zhu2019generalized,zhu2021effective,zhu2021hypoellipticity3,lei2016data,chu2017mori}.  
\paragraph{\textbf{FDRs for other transport quantities}} The extended SNS equation were studied in Section \ref{sec:1st_FDR} for the study of turbulent dispersion. Generally speaking, because the SNS equation is a {\em Eulerian} interpretation of the turbulent velocity (or vorticity) field, if the observable $O(t)$ under the investigation is a {\em Lagrangian} variable, then similar extended dynamics is needed to study the corresponding transport phenomenon. For the same reason, the SNS equation itself would be sufficient for study of the {\em Eulerian} observables. For instance, the dynamics of the Fourier mode kinetic energy $u^2_k(t)$ is closely related to the eddy viscosity which characterizes the transport of the turbulent energy in different space scales \cite{kraichnan1976eddy}. Using the first FDR \eqref{heat_linear_response} for the observable $O(t)=\alpha_k^2(t),\beta_k^2(t)$ of a 2D SNS equation, we obtain:
\begin{align}
    \langle \alpha_k^2(t)\rangle_{\rho_{\delta}}-\langle \alpha_k^2(t)\rangle_{\rho}=\int_0^t R_{\alpha_k^2,\Delta\omega}(t,s)ds,
    \qquad
     \langle \beta_k^2(t)\rangle_{\rho_{\delta}}-\langle \beta_k^2(t)\rangle_{\rho}=\int_0^t R_{\beta_k^2,\Delta\omega}(t,s)ds,
\end{align}
where 
\begin{align}\label{Mobility_X(t)_3}
   R_{\alpha_k^2,\Delta\omega}(t,s)
   =\langle \mathcal{R}(t)\alpha_k^2(s)\rangle_{\rho'},
   \qquad
   R_{\beta_k^2,\Delta\omega}(t,s)
   =\langle \mathcal{R}(t)\beta_k^2(s)\rangle_{\rho'},
\end{align}
and $\mathcal{R}(t)$ is as \eqref{R(s)_fluid_general}, $\langle\cdot\rangle_{\rho'}$ is the path-space ensemble average with respect to the NESS of \eqref{eqn:alpha_beta0}. The nonlinear response formula can be obtained similarly with the proper defined response function $N(t)$ \eqref{R(s)_nonlinear_response1}, as well as the generalized second FDR of the form \eqref{general_2nd_FDT}. 
\subsection{Comparison with the renormalized perturbation theory}\label{sec:compare}
The formalism introduced in Section \ref{sec:1st_FDR}-\ref{sec:2nd_FDR} can be termed as the {\em nonequilibrium perturbation theory} for turbulence.  It is instructive to compare it with the established {\em renormalized perturbation theory} \cite{mccomb1990physics} for turbulence and assess their differences and connections. Roughly speaking, the renormalized perturbation approach for turbulence is a cluster of methods to study the NS (or SNS) equation using tools such as the renormalized perturbation expansion and renormalization group analysis which are borrowed from statistical field theory \cite{parisi1988statistical} or quantum filed theory \cite{Justin}. One of the first renormalized perturbation theories is the Direct Interaction Approximation (DIA) method developed by Kraichnan \cite{kraichnan1959structure}. \footnote{It was pointed out in \cite{mccomb1990physics} that the DIA method is actually a special type of the renormalized perturbation theory.} Over the years, many researchers contributed to its development and one may refer to McComb's monograph \cite{mccomb1990physics} and the reference therein for in-depth explorations. In brief, the discussion we are going to present indicates that the nonequilibrium perturbation theory for turbulence is {\em rather different} from the renormalized perturbation theory in terms of the methodologies and the key questions the theory aims to answer. In fact, we believe the new approach provides a {\em paradigm shift} on the turbulence study using statistical approaches. Based on the new formalism, many difficulties encountered in the renormalized theory can be simply avoided. In this section, we will adopt the setting and notations used in  Chapter 5-6 \cite{mccomb1990physics} for briefly reviewing the main structure of the renormalized perturbation theory and then make the comparison. To this end, we consider with the SNS equation for the velocity field in torus $\T^d=[-\pi,\pi]^d$:
\begin{align}\label{eqn:3D_NS_v}
    (\partial_t+\nu |k|^2)u_{\alpha}(k,t)=\lambda M_{\alpha\beta\gamma}(k)\sum_ju_{\beta}(j,t)u_{\gamma}(k-j,t)+D_{\alpha\beta}(k)f_{\beta}(k,t),
\end{align}
where $u_{\alpha}(k,t)$ the Fourier mode of the velocity field in $\alpha$-direction with wavenumber $k$. $\lambda$ is a auxiliary constant which is added to facilitate the perturbation analysis. Note that the summation convention for tensors is employed. Eqn \eqref{eqn:3D_NS_v} can represent 2D ($d=2$) or 3D ($d=3$) SNS equation. The common setting for the stochastic force is assumed the Brownian motion with correlation function: 
\begin{align}\label{f_corre}
\langle f_{\alpha}(k,t)f_{\beta}(k,t)\rangle=D_{\alpha\beta}(k)w(|k|,t-t'),
\end{align}
where $w(|k|,t-t')$ can be specified for different modeling purposes. The goal of the renormalized perturbation theory is find a self-consistent way to calculate two statistical quantities: the Green function $G_{\alpha\beta}(k;t,t')$ and the velocity correlation function $Q_{\alpha,
\beta}(k;t,t')$:
\begin{align}\label{eqn:G_Q}
\langle\delta u_{\alpha}(k,t)\rangle=\int_{-\infty}^t G_{\alpha\beta}(k;t,t')\delta f_{\beta}(k,t')dt'
,\qquad 
\langle u_{\alpha}(k,t)u_{\beta}(k',t')\rangle=
\delta_{k+k',0}Q_{\alpha\beta}(k;t,t'),
\end{align}
or their correspondence for isotropic turbulence:
\begin{align}\label{eqn:G_Q_iso}
\langle\delta u_{\alpha}(k,t)\rangle=\int_{-\infty}^t D_{\alpha\beta} G(|k|;t,t')\delta f_{\beta}(k,t')dt'
,\qquad 
\langle u_{\alpha}(k,t)u_{\beta}(k',t')\rangle=
D_{\alpha\beta}\delta_{k+k',0}Q(|k|;t,t').
\end{align}
Here $\delta f_{\beta}(k,t)$ is the perturbation of the stochastic force. From the definitions \eqref{eqn:G_Q}-\eqref{eqn:G_Q_iso}, one can see that $G(|k|;t,t')$ is the linear response function of the fluid velocity with respect to the perturbation of the external force and $Q(|k|;t,t')$ is the time-autocorrelation function of the velocity field. To obtain $G(|k|;t,t')$ and $Q(|k|;t,t')$, the renormalized perturbation theory considers the following perturbation series with respect to the auxiliary parameter $\lambda$:
\begin{align}
    u_{\alpha}(k,t)&=u_{\alpha}^{(0)}(k,t)+\lambda u_{\alpha}^{(0)}(k,t)+ O(\lambda^2)\label{eqn:u_lambda_pert}\\
    G_{\alpha\beta}(k;t,t')&=G_{\alpha\beta}^{(0)}(k;t,t')+\lambda G_{\alpha\beta}^{(1)}(k;t,t')+ O(\lambda^2)\label{eqn:G_lambda_pert}\\
    Q_{\alpha\beta}(k;t,t')&=Q_{\alpha\beta}^{(0)}(k;t,t')+\lambda Q_{\alpha\beta}^{(1)}(k;t,t')+ O(\lambda^2)\label{eqn:Q_lambda_pert}.
\end{align}
In Eqn \eqref{eqn:u_lambda_pert}-\eqref{eqn:Q_lambda_pert}, $u_{\alpha}^{(0)}(k,t)$, $G_{\alpha\beta}^{(0)}(k;t,t')$ and $Q_{\alpha\beta}^{(0)}(k;t,t')$ are the zero-order velocity field, Green function and the velocity correlation function. They come from the {\em linearization} of the NS equation \eqref{eqn:3D_NS_v}. In particular,
\begin{align}\label{eqn:lambda_pert}
    (\partial_t+\nu |k|^2)u_{\alpha}^{(0)}(k,t)&=D_{\alpha\beta}f_{\beta}(k,t)\\
    (\partial_t+\nu |k|^2)G_{\alpha\beta}^{(0)}(k;t,t')&=D_{\alpha\beta}\delta(t-t')\\
    \langle u_{\alpha}^{(0)}(k-j,t),u_{\beta}^{(0)}(l,t')\rangle&=\delta_{k-j+l,0}Q_{\alpha\beta}^{(0)}(k-j;t,t').
\end{align}
By plugging \eqref{eqn:u_lambda_pert} into the nonlinear SNS equation \eqref{eqn:3D_NS_v} and matching terms with the same power of $\lambda$, one can get the exact expressions for the power series \eqref{eqn:G_lambda_pert} and \eqref{eqn:Q_lambda_pert}. This procedure leads to the {\em primitive perturbation expansion} for $G_{\alpha\beta}(k;t,t')$ and $Q_{\alpha\beta}(k,t,t')$. The key feature of this expression is that $G_{\alpha\beta}(k;t,t')$ and $Q_{\alpha\beta}(k,t,t')$ are represented as a complicated summation with building blocks $G_{\alpha\beta}^{(0)}(k;t,t')$ and $Q_{\alpha\beta}^{(0)}(k;t,t')$. This perturbation series has Feynman diagram representation which can be {\em renormalized} using different methods by grouping sub-diagrams of the total summation (see details in Chapter 5-6 \cite{mccomb1990physics}). As a result, the right hand side of the renormalized expansion series has terms such as $G_{\alpha\beta}(k;t,t')$ and $Q_{\alpha\beta}(k;t,t')$ which makes \eqref{eqn:G_lambda_pert}-\eqref{eqn:Q_lambda_pert} self-consistent equations for $G_{\alpha\beta}(k;t,t')$ and $Q_{\alpha\beta}(k;t,t')$. Karichnan's DIA theory is a second-order renormalized perturbation theory (i.e. truncation at the order $O(\lambda^2)$) with the line renormalization which leads to the following coupled evolution equation for $G(|k|;t,t')$ and $Q(|k|;t,t')$ of the isotropic turbulence:
\begin{equation}\label{eqn:DIA}
\begin{aligned}
   (\partial_t+\nu |k|^2)G(|k|;t,t')&=-\sum_j\int_{t'}^tdsL(k,j)G(j;t,s)G(k;s,t')Q(|k-j|;t,s)+\delta(t-s)\\
   (\partial_t+\nu |k|^2)Q(|k|;t,t')&=\sum_jL(k,j)\bigg\{\int_{-{\infty}}^{t'}dsG(|k|;t',s)Q(|j|;t,s)Q(|k-j|;t,s)\\
   &-\int_{-\infty}^tdsG(|j|;t',s)Q(|k-j|;t,s)Q(k;t',s)\bigg\}+\int_{-\infty}^{t'}G(|k|;t',s)w(|k|;t,s)ds,
\end{aligned}
\end{equation}
where $w(|k|;t,s)$ is defined by \eqref{f_corre} and 
\begin{align}
L(k,j)=-2M_{\epsilon\beta\gamma}(k)M_{\beta\delta\epsilon}(j)D_{\gamma\delta}(k-j).
\end{align}
To sum up, the main procedure of the {\em renormalized perturbation theory} can be schematically represented as Figure \ref{fig:perturbation}: For (S)NS with different viscosity $\nu=1/Re$, the theory uses the linearized (S)NS equation to as the starting point and invokes the truncated renormalized series to get the information,  i.e. $G_{\alpha\beta}(k;t,t')$ and $Q_{\alpha\beta}(k;t,t')$, for the nonlinear (S)NS equation. In the process, the auxiliary parameter $\lambda$ is added at the beginning in order to facilitate the perturbation analysis but is taken to be $\lambda=1$ in the end. Hence from an information theoretic point of view, the renormalized perturbation theory uses the linearized (S)NS equation, which can be solved {\em exactly}, to predict the dynamics of the nonlinear (S)NS equation. 

In Section \ref{sec:1st_FDR}, we have shown that the linear response function $G_{\alpha\beta}(k;t,t')$ and the velocity correlation function $Q_{\alpha\beta}(k;t,t')$ can be obtained using the operator-form first FDR or the path-integral-form first FDR.
\footnote{Note that what we derived in Section \ref{sec:1st_FDR} is the response quantity $\mathcal{R}(s)$ or $\N(t)$ with respect to perturbation $+\delta\mu\Delta \omega$, and the correlation function is for the Fourier mode of the vorticity field. Using the functional Radon-Nikodym derivative method provided in \ref{app:Radon-Nikodym_proof}, it is direct to get the response formula corresponding to the noise perturbation $\delta f_{\beta}(k,t)$. On the other hand, using the relationship between the stream function $\psi(x,t)$, the velocity field $u(x,t)$ and the vorticity field $\omega(x,t)$, we can get the velocity time-autocorrelation function $Q_{\alpha\beta}(k;t,t')$ from the vorticity time-autocorrelation function. These results can be readily generalized to the 3D turbulence.} However, the paradigm for the {\em nonequilibrium perturbation theory} is very different. From Figure \ref{fig:perturbation}, one can see that the starting point of the perturbation theory is the NESS of the {\em nonlinear} (S)NS equation. Then we used the response theory to get the information about {\em another} NESS of the nonlinear (S)NS equation which deviates from the first one with added perturbations such as $+\delta\mu\Delta \omega$. From an information theoretic point of view, the nonequilibrium perturbation theory uses the nonlinear (S)NS equation, which {\em cannot} be solved exactly but an approximated solution can be obtained via numerical simulations, to predict the dynamics of another nonlinear (S)NS equation. We admit that the renormalized perturbation theory tries to answer some more difficult, more fundamental questions such as the formation of the Kolmogorov's 5/3 energy spectrum for turbulent flows (see Chapter 8 in \cite{mccomb1990physics}). Such questions cannot be answered using the nonequilibrium perturbation theory framework. Instead, one can answer practical questions such as: given the dynamics of a turbulent field with a certain Reynolds number $Re$, what are the Lagrangian particle diffusion rate and the eddy energy diffusion rate for fluid with Reynolds number $Re+\delta Re$.

The theoretical convergence for these two perturbation theories are also different. Specifically, since for highly turbulent flow with small viscosity $\nu=1/Re$, the nonlinear terms are actually dominant in the dynamical process. One has to use the renormalized perturbation series, or renormalization groups to construct the dynamical models if the starting point of the analysis is the linearized NS equation. From a theoretical point view, the accuracy and convergence of such perturbation analysis cannot be guaranteed even after the renormalization since after all the auxiliary parameter $\lambda=1$. However, the convergence for nonequilibrium perturbation theory {\em can} be guaranteed since the nonlinear response result \eqref{R(s)_nonlinear_response}-\eqref{R(s)_nonlinear_response1} is {\em exact}, and the linear response approximation \eqref{R(s)_fluid_general} is valid since normally we choose $\delta\mu\ll 1$. In addition, renormalization is not needed. From the field theory point of view, one may view the path-integral approach as a functional method to derive the nonequilibrium perturbation theory, while the renormalized perturbation series used a method that is similar to the canonical quantization.

Lastly, we point out that the nonequilibrium perturbation theory provides a simple, flexible framework to study the turbulent transport for different fluid equations. From the discussion of Section \ref{sec:1st_FDR}, one can see that it generally applies to non-isotropic turbulence and different turbulent transport processes can be handled in a unified framework. The nonequilibrium perturbation theory can also be combined with the renormalized perturbation theory. Namely, we may use the renormalized perturbation theory first to get a detailed evaluation for quantities such as the $G_{\alpha\beta}(k;t,t')$ and $Q_{\alpha\beta}(k;t,t')$. Based on which, the nonequilibrium perturbation theory is then applied to get linear response results  $G_{\alpha\beta}^{\delta}(k;t,t')$ and $Q_{\alpha\beta}^{\delta}(k;t,t')$ corresponding to a perturbed turbulence in the NESS. 
\begin{figure}[t]
\centerline{
\includegraphics[height=3.5cm]{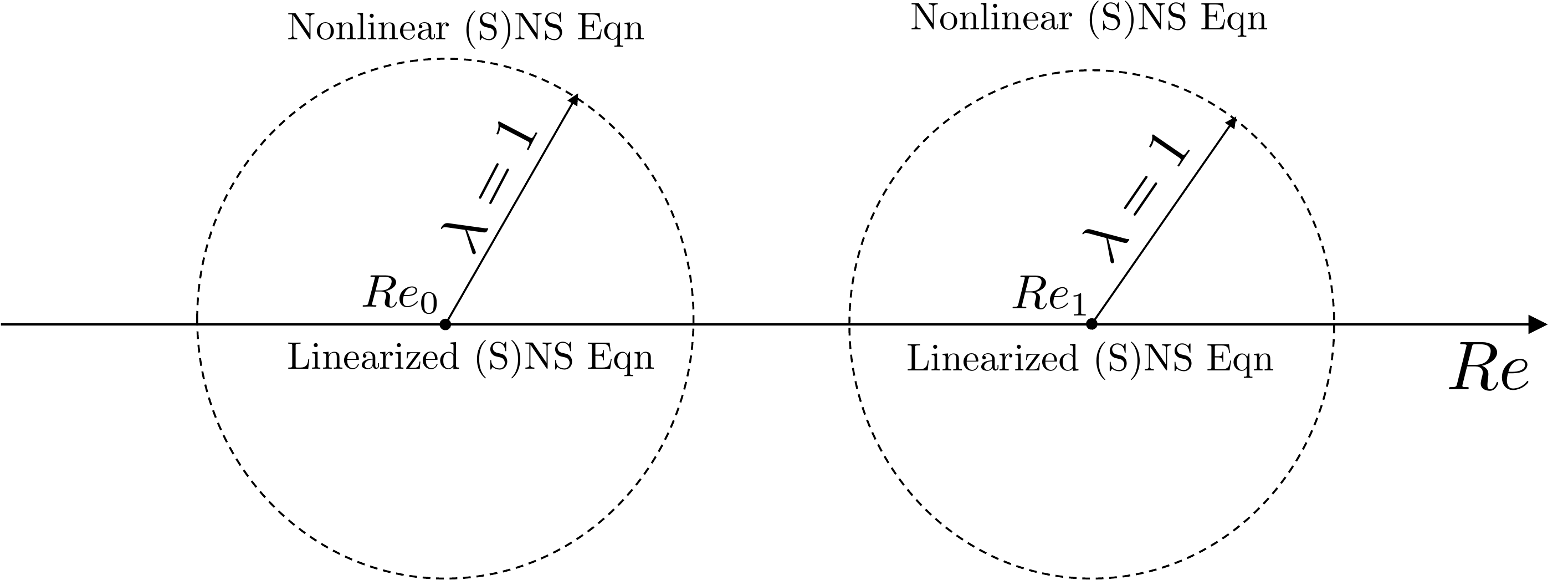}
}
\vspace{0.5cm}
\centerline{
\includegraphics[height=3.5cm]{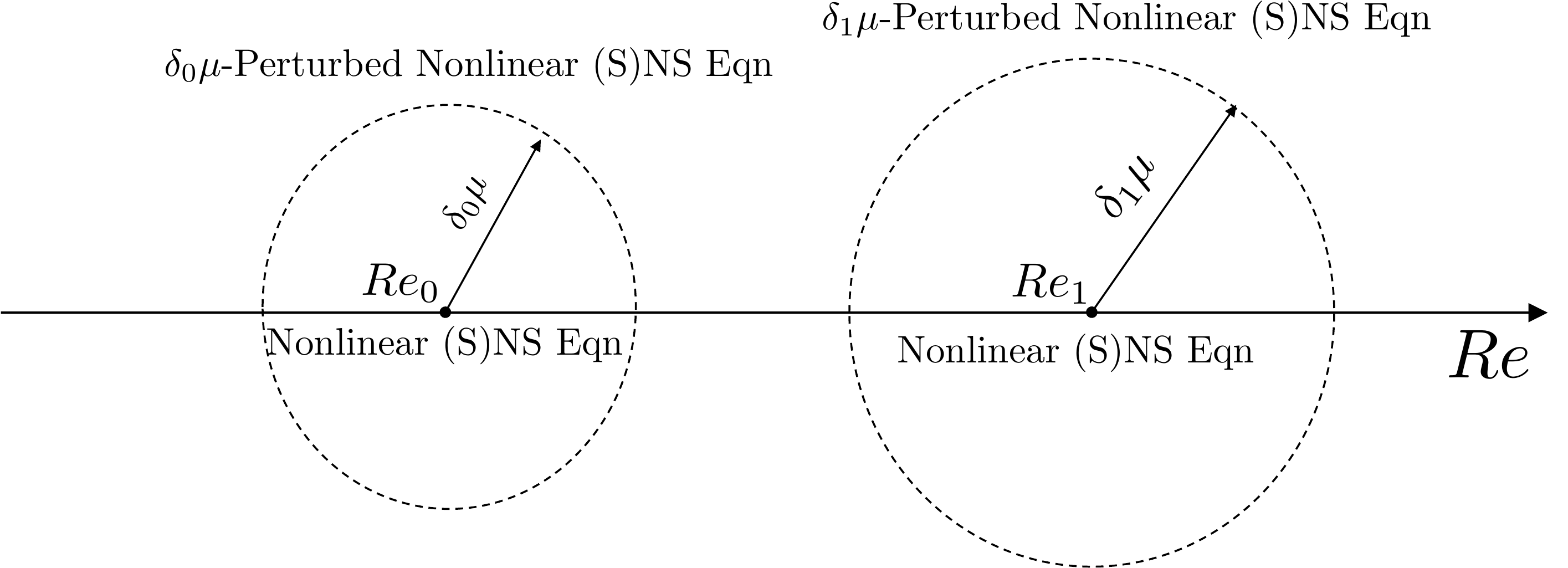}
}
\caption{Schematic comparison of the renormalized perturbation theory (Top) and the nonequilibrium perturbation theory (Bottom).}
\label{fig:perturbation} 
\end{figure}

\section{Conclusions}\label{sec:summary}
In this paper, we combined Kraichnan's turbulence thermalization assumption with Ruelle's recent work on the heat transport analogy of turbulence, and proposed a nonequilibrium statistical mechanics framework to study turbulent transport. We demonstrated that the viscous fluid equation in the frequency domain can be understood as the nonequilibrium heat conduction model where the Fourier modes are regarded as the ``particles''. As an application of this framework, we studied in detail the turbulent dispersion problem for the 2D stochastic Navier-Stokes equation and proved several generalized fluctuation-dissipation relations which are valid in the nonequilibrium. In particular, we obtained a nonequilibrium Kubo formula and a nonlinear response formula via the path-integral technique, which can be used directly to predict the turbulent dispersion in a perturbed nonequilibrium steady state. Using the Mori-Zwanzig framework, we proved a generalized second fluctuation-dissipation relation and built a closed reduced-order model for the Lagrangian particle with the resulting generalized Langevin equation. Further applications of the proposed framework to other fluid systems and other turbulent transport phenomena can be expected. 

\vspace{0.3cm}
\noindent 
{\bf Acknowledgements} 
We would like to thank Prof. Fei Lu for stimulating discussions. We also thank Prof. Xiantao Li and Dr. Lorenzo Caprini for bringing to my attention the recent progress of path-integral-form first FDR, which helps to form Section \ref{sec:1st_FDR_path} of this paper.  

\appendix
\section{Invariant measure for the 2D SNS equation}
\label{app:Enstrophy_measure_proof}
This result was proved by Albeverio et al \cite{albeverio1990global} for a special case (enstrophy-canonical measure) in a slightly different setting. For the generalized case considered in this paper, the proof can be briefly summarized as follow: Since $\Omega_{p}$, $p\geq 0$ are the invariants of the 2D Euler equation \eqref{eqn:Euler}, naturally \eqref{inv_measure_enstrophy} is an invariant measure for the Euler flow corresponding to Eqn \eqref{eqn:Euler} and 
\begin{align}\label{eqn:oper_Iden}
    \partial_t\rho_{eq,p}'=\L^*\rho_{eq,p}'=\sum_l\frac{\partial}{\partial\alpha_l}(F_l\rho_{eq,p}')+\frac{\partial}{\partial\beta_l}(G_l\rho_{eq,p}')=0,
\end{align}
where $\L^*$ is the adjoint of the Liouville operator corresponding to the Euler flow. On the other hand, the Kolmogorov forward operator $\K^*$ corresponding to the 2D SNS \eqref{eqn:alpha_beta_simplified} is given by
\begin{align}\label{K_SNS}
\K^*=\sum_l\nu|l|^2\frac{\partial}{\partial\alpha_l}(\alpha_l\cdot)+\nu|l|^2\frac{\partial}{\partial\beta_l}(\beta_l\cdot)
+\frac{\partial}{\partial\alpha_l}(F_l\cdot)+\frac{\partial}{\partial\beta_l}(G_l\cdot)+\frac{1}{2}\sigma_l^2\frac{\partial^2}{\partial\alpha_l^2}+\frac{1}{2}\gamma_l^2\frac{\partial^2}{\partial\beta_l^2}.
\end{align}
When $\sigma_l^2=\gamma_l^2=|l|^{4-2p}$, $\nu=\nu_p$, we can rewrite $\K^*\rho_{eq,p}'$ as
\begin{equation*}
\begin{aligned}
\K^*\rho_{eq,p}'&=\sum_l\nu_{p}|l|^2\frac{\partial}{\partial\alpha_l}(\alpha_l\rho_{eq,p}')+\nu_p|l|^2\frac{\partial}{\partial\beta_l}(\beta_l\rho_{eq,p}')\\
&+\frac{\partial}{\partial\alpha_l}(F_l\rho_{eq,p}')+\frac{\partial}{\partial\beta_l}(G_l\rho_{eq,p}')+\frac{1}{2}\sum_l|l|^{4-2p}\left(\frac{\partial^2\rho_{eq,p}'}{\partial\alpha_l^2}+\frac{\partial^2\rho_{eq,p}'}{\partial\beta_l^2}\right)\\
&=\L^*\rho_{eq,p}'+\sum_{l}\K_{l}^*\rho_{eq,p}'
\propto\sum_{l}\K_{l}^*\exp\left\{-\nu_p\Omega_{2p}\right\}
=\sum_{l}\K_{l}^*\exp\left\{-\nu_p\sum_l|l|^{2p-2}(\alpha_l^2+\beta_l^2)\right\},
\end{aligned}
\end{equation*}
where 
\begin{equation*}
\begin{aligned}
\K_{l}^*=\nu_p|l|^2\frac{\partial}{\partial\alpha_l}(\alpha_l\cdot)+\nu_p|l|^2\frac{\partial}{\partial\beta_l}(\beta_l\cdot)+\frac{1}{2}|l|^{4-2p}\left(\frac{\partial^2}{\partial\alpha_l^2}+\frac{\partial^2}{\partial\beta_l^2}\right).
\end{aligned}
\end{equation*}
Then we have 
\begin{equation*}
\begin{aligned}
\K^*\rho_{eq,p}'=\sum_{l}\K^*_l\rho_{eq,p}'&=\sum_{l}\nu_p|l|^2\frac{\partial}{\partial\alpha_l}(\alpha_l\rho_{eq,p}')-\sum_l\nu_p|l|^{4-2p+2p-2}\frac{\partial}{\partial\alpha_l}(\alpha_l\rho_{eq,p}')
\\
&+\sum_lv_p|l|^2\frac{\partial}{\partial\beta_l}(\beta_l\rho_{eq,p}')-\sum_lv_p|l|^{4-2p+2p-2}\frac{\partial}{\partial\beta_l}(\beta_l\rho_{eq,p}')\\
&=0.
\end{aligned}
\end{equation*}
Therefore \eqref{inv_measure_enstrophy} is also the invariant measure for the 2D SNS equation \eqref{eqn:alpha_beta_simplified}. In fact, as pointed out by Albeverio et al. \cite{albeverio1990global}, the 2D SNS equation in \textbf{\eqref{FDR_perturbation}} can be viewed as the 2D Euler equation \eqref{eqn:Euler} plus a infinite-dimensional Orstein-Ulenbeck (OU) process. This OU process has generator $\sum_{l}\K_l$ and shares the invariant measure with the original Euler dynamics, therefore \eqref{inv_measure_enstrophy} indeed is the invariant measure of the 2D SNS equation. The enstrophy-type invariant measure $\rho_{eq,1}'$ is often used in the modeling of equilibrium state of 2D Euler flow \cite{kraichnan1980two}. To be noticed that it implies the enstrophy equipartition for each $\Omega_{2,k}$, where $\Omega_{2,k}$ is the enstrophy in $k$-th Fourier mode. 

%
%
\section{The exact response formula for the 2D SNS equation}\label{app:Radon-Nikodym_proof}
In this section, we prove the path-integral-form response formulas for the 2D SNS equation driven by the non-degenerate noise. For a $(N+1)^2+1$-dimensional stochastic system \eqref{eqn:alpha_beta_simplified}, note that if the imposed Gaussian white noises $\sigma_ld\B_l,\gamma_ld\W_l,\kappa_1d\D_1,\kappa_2d\D_2$ are non-degenerate ($\sigma_l,\gamma_l,\kappa_1,\kappa_2>0$) and independent with each other. Then applying Stratonovich's formula (2.4.42) in \cite{moss1989noise}, we can get the explicit expression of the path-integral-form functional probability density $P[x(t)]$\footnote{Since for our case $h(\tilde x_n)=h$ is a constant matrix (see \eqref{matrices_stra}), the late term in the exponential part of the Stratonovich's formula (2.4.42) \cite{moss1989noise} vanishes.}:
\begin{align}\label{p(x(t))_stratonovich}
    P[x(t)]\propto\prod_{n=1}^{(N+1)^2+1}\text{det}^{-1}\|h(\tilde x_n)\|\cdot\exp\left\{-\frac{1}{2}\int_0^t\sum_{n=1}^{(N+1)^2+1}
    b_{ii}^{-1}(\dot{x}_i-g_i)^2+\frac{\partial}{\partial x_i}g_ids\right\},
\end{align}
where $b_{ii}^{-1}$ is the diagonal element of the inverse of matrix $b$ (given below). Each term in \eqref{p(x(t))_stratonovich} is given by
\begin{equation}
\begin{aligned}\label{matrices_stra}
h(\tilde x_n)=h=
\begin{bmatrix}
[\sigma_l] &  & &\\
& [\gamma_l]  & &\\
&  & \kappa_1&\\
&  & &\kappa_2
\end{bmatrix},
\quad
b=
\begin{bmatrix}
[\sigma_l^2] &  & &\\
& [\gamma_l^2]  & &\\
&  & \kappa_1^2&\\
&  & &\kappa_2^2
\end{bmatrix}
\quad
x=
\begin{bmatrix}
[\alpha_l]\\
[\beta_l]\\
X_1\\
X_2
\end{bmatrix}
\quad
g=
\begin{bmatrix}
[F_l-\nu|l|^2\alpha_l]\\
[G_l-\nu|l|^2\beta_l]\\
\partial_{X_2}\psi(X,t)\\
-\partial_{X_1}\psi(X,t)
\end{bmatrix},
\end{aligned}
\end{equation}
where $[\cdot]$ means block matrix. Note that $h$ and $b$ are constant, diagonal matrices, with \eqref{p(x(t))_stratonovich}-\eqref{matrices_stra}, for 2D SNS equation, we can write its functional probability density as: 
\begin{equation}\label{p(x(t))_stratonovich_exp}
\begin{aligned}
    P[x(t)]&\propto\exp\bigg\{-\frac{1}{2}\int_0^tds
    \sum_{l}\frac{1}{\sigma_l^2}[\dot{\alpha}_l-(F_l-\nu|l|^2\alpha_l)]^2
    +\sum_{l}\frac{1}{\gamma_l^2}[\dot{\beta}_l-(G_l-\nu|l|^2\beta_l)]^2\\
    &
    +\frac{1}{\kappa_1^2}(\dot{X}_1-\partial_{X_2}\psi(X,s))^2
    +\frac{1}{\kappa_2^2}(\dot{X}_2+\partial_{X_1}\psi(X,s))^2
    +\sum_l\partial_{\alpha_l}(F_l-\nu|l|^2\alpha_l)
    +\sum_l\partial_{\beta_l}(G_l-\nu|l|^2\beta_l)
  \bigg\}.
\end{aligned}
\end{equation}
In \eqref{p(x(t))_stratonovich_exp}, we eliminated the following term since $\partial_{X_1}\partial_{X_2}\psi(X,t)-\partial_{X_2}\partial_{X_1}\psi(X,t)=0$. Using the exact same method to derive the functional probability density for $P[x^{\delta}(t)]$, where the perturbed system is given by \eqref{FDR_perturbation}, and combining the Radon-Nikodym derivative \eqref{Radon-Nikodym_derivative}, we can get that 
\begin{equation}\label{RN-fluid}
\begin{aligned}
    \frac{P[x^{\delta}(t)]}{P[x(t)]}
    &\propto\exp\left\{\frac{1}{2}S(t)-\frac{1}{2}\mathcal{T}(t)\right\}\\
    &\propto\exp\bigg\{-\frac{1}{2}\int_0^tds\sum_l\left(\frac{2\delta\mu|l|^2}{\sigma_l}\alpha_lB_l
    +\frac{2\delta\mu|l|^2}{\gamma_l}\beta_lW_l
    +\frac{\delta^2\mu^2|l|^4}{\sigma_l^2}\alpha^2_l
    +\frac{\delta^2\mu^2|l|^4}{\beta_l^2}\gamma^2_l
    -2\delta\mu|l|^2
    \right)\bigg\}\\
    &\propto
    \exp\bigg\{-\int_0^tds\sum_l\left(\frac{\delta\mu|l|^2}{\sigma_l}\alpha_lB_l
    +\frac{\delta\mu|l|^2}{\gamma_l}\beta_lW_l
    +\frac{\delta^2\mu^2|l|^4}{2\sigma_l^2}\alpha^2_l
    +\frac{\delta^2\mu^2|l|^4}{2\beta_l^2}\gamma^2_l
    \right)\bigg\},
\end{aligned}
\end{equation}
where in the second step, we used the fact that $\dot{\alpha_l}-(F_l-\nu|l|^2\alpha_l)=\sigma_ld\B_l=\sigma_lB_lds$ and $\dot{\beta_l}-(G_l-\nu|l|^2\beta_l)=\gamma_ld\W_l=\gamma_lW_lds$. In the third step, since the last term in the exponential is a constant, hence was merged into the normalization constant. With \eqref{RN-fluid} and the path-space ensemble average formula \eqref{ensemble_av1}, it is easy to get the linear response formula \eqref{R(t)_Omega_2p}. Similarly, using formula \eqref{ensemble_av}, we obtain the nonlinear response relation \eqref{R(s)_nonlinear_response}.

We also consider a generalized case which is more comparable with the renormalized perturbation theory since over there the noise can be colored. Specifically, we allow the imposed Gaussian noise $\sigma_ld\B_l$, $\gamma_ld\W_l$, $\kappa_1d\D_1,\kappa_2d\D_2$ to be independent with each other but can be colored in time with correlation function:
\begin{align*}
    \langle B_l(t),B_l(t')\rangle&=B'_l(t,t'),\qquad \text{where}\qquad d\B_l(t)=B_l(t)dt\\
    \langle W_l(t),W_l(t')\rangle&=W'_l(t,t'),\qquad \text{where}\qquad d\W_l(t)=W_l(t)dt\\
    \langle D_1(t),D_1(t')\rangle&=D'_1(t,t'),\qquad \text{where}\qquad d\D_1(t)=D_1(t)dt\\
    \langle D_l(t),D_l(t')\rangle&=D'_2(t,t'),\qquad \text{where} \qquad
    d\D_2(t)=D_2(t)dt.
\end{align*}
Then using Stratonovich's formula (2.4.37) in \cite{moss1989noise}, we can express the functional probability $P[x(t)]$ as
\begin{equation*}
\begin{aligned}
    P[x(t)]\propto\exp\bigg\{&-\frac{1}{2}\int_0^t\int_0^tdsds'
    \sum_{l}\frac{B^{'-1}_l(s,s')}{\sigma_l^2}[\dot{\alpha}_l(s)-(F_l(s)-\nu|l|^2\alpha_l(s))][\dot{\alpha}_l(s')-(F_l(s')-\nu|l|^2\alpha_l(s'))]\\
    &+\sum_{l}\frac{W^{'-1}_l(s,s')}{\gamma_l^2}[\dot{\beta}_l(s)-(G_l(s)-\nu|l|^2\beta_l(s))]
    [\dot{\beta}_l(s')-(G_l(s')-\nu|l|^2\beta_l(s'))]
    \\
    &
    +\frac{D^{'-1}_1(s,s')}{\kappa_1^2}(\dot{X}_1(s)-\partial_{X_2}\psi(X,s))
    (\dot{X}_1(s')-\partial_{X_2}\psi(X,s'))\\
    &+\frac{D^{'-1}_2(s,s')}{\kappa_2^2}(\dot{X}_2(s)+\partial_{X_1}\psi(X,s))
    (\dot{X}_2(s')+\partial_{X_1}\psi(X,s'))\\
    &-\frac{1}{2}\int_0^tds
    \sum_l\partial_{\alpha_l}(F_l-\nu|l|^2\alpha_l)
    +\sum_l\partial_{\beta_l}(G_l-\nu|l|^2\beta_l)
  \bigg\},
\end{aligned}
\end{equation*}
where the functions such as $B^{'-1}_l(s,s')$ are the ``inverse kernel'' defined by
\begin{equation}
\begin{aligned}
\int_0^T B^{'-1}_l(s,s')u(s')ds'=v(s) \quad \Leftrightarrow \quad 
\int_0^T B'_l(s,s')v(s')dt'=u(s) \\
\int_0^T W^{'-1}_l(s,s')u(s')ds'=v(s) \quad \Leftrightarrow \quad 
\int_0^T W'_l(s,s')v(s')ds'=u(s) \\
\int_0^T D^{'-1}_1(s,s')u(s')dt'=v(s) \quad \Leftrightarrow \quad 
\int_0^T D'_1(s,s')v(s')ds'=u(s) \\
\int_0^T D^{'-1}_2(s,s')u(t')dt'=v(s) \quad \Leftrightarrow \quad 
\int_0^T D'_2(s,s')v(s')dt'=u(s),
\end{aligned}
\end{equation}
where $u(s)$ is an arbitrary function. Using the Radon-Nikodym derivative \eqref{Radon-Nikodym_derivative}, we obtain 
\begin{equation}\label{RN-fluid_correlation}
\begin{aligned}
    \frac{P[x^{\delta}(t)]}{P[x(t)]}
    &\propto\exp\left\{\frac{1}{2}S(t)-\frac{1}{2}\mathcal{T}(t)\right\}\\
    &\propto\exp\bigg\{-\frac{1}{2}\int_0^t\int_0^tdsds'
    \sum_l\frac{\delta\mu|l|^2B^{'-1}(s,s')}{\sigma_l}[\alpha_l(s)B_l(s')+\alpha_l(s')B_l(s)]\\
    &
    +\sum_l\frac{\delta\mu|l|^2W^{'-1}(s,s')}{\gamma_l}[\beta_l(s)W_l(s')+\beta_l(s')W_l(s)]\\
    &
    +\sum_l\frac{\delta^2\mu^2|l|^4B^{'-1}(s,s')}{\sigma_l^2}[\alpha_l(s)\alpha_l(s')]
    +\sum_l\frac{\delta^2\mu^2|l|^4W^{'-1}(s,s')}{\gamma_l^2}[\beta_l(s)\beta_l(s')]
    \bigg\}.
\end{aligned}
\end{equation}
Then using the functional derivative, we obtain
\begin{equation}
\begin{aligned}
\mathcal{R}(s)=\mathcal{R}(s,t)=\int_0^t ds'&\sum_l\frac{\delta\mu|l|^2B^{'-1}(s,s')}{\sigma_l}[\alpha_l(s)B_l(s')+\alpha_l(s')B_l(s)]\\
    &+\sum_l\frac{\delta\mu|l|^2W^{'-1}(s,s')}{\gamma_l}[\beta_l(s)W_l(s')+\beta_l(s')W_l(s)],
\end{aligned}
\end{equation}
and the nonlinear response function $\N(t)$ is just the right hand side of \eqref{RN-fluid_correlation}.
\section{The tentative response formula for deterministic turbulence}
\label{app:determ_turbulence_FDR}
We consider the following perturbation of the 2D SNS equation:
\begin{align}\label{FDR_perturbation_deter}
    \begin{dcases}
    d\alpha_l&=F_{l}dt-\nu|l|^2\alpha_ldt+\eta f_l+\sigma_ld\B_l\\
    d\beta_l&=G_{l}dt-\nu|l|^2\beta_ldt+\theta g_l+\gamma_ld\W_l\\
    dX_1&=\partial_{X_2}\psi(X,t)dt+\kappa_1d\D_1(t)\\
    dX_2&=-\partial_{X_1}\psi(X,t)dt+\kappa_2d\D_2(t)
    \end{dcases}
    \quad
\xRightarrow[]{-\delta\mu\Delta \omega}
\quad
    \begin{dcases}
    d\alpha_l&=F_{l}dt-\nu'|l|^2\alpha_ldt+\eta f_l+\sigma_ld\B_l\\
    d\beta_l&=G_{l}dt-\nu'|l|^2\beta_ldt+\theta g_l+\gamma_ld\W_l\\
    dX_1&=\partial_{X_2}\psi(X,t)dt+\kappa_1d\D_1(t)\\
    dX_2&=-\partial_{X_1}\psi(X,t)dt+\kappa_2d\D_2(t)
    \end{dcases}.
\end{align}
Note that in \eqref{FDR_perturbation_deter}, we added deterministic forces $\eta f_l$ and $\theta g_l$ in the Fourier modes in order to discuss the case of NESS3. To get the linear response formula for deterministic turbulence, we note that the deterministic NS equation for NESS1 and NESS3 is the limit of 
\eqref{FDR_perturbation_deter} as $\sigma_l,\gamma_l\rightarrow 0$. This limits is singular in the sense of \cite{matkowsky1977exit} because the stationary Fokker-Planck equation (Eqn (2.16)-(2.17) in \cite{matkowsky1977exit}) will degenerate from a second-order elliptic equation to a first-order hyperbolic equation. To get the linear response formula for such cases, we will first consider a stochastic system perturbation \eqref{FDR_perturbation_deter} using the aforementioned path-integral technique, then by taking the limit $\sigma_l,\gamma_l\rightarrow 0$ and getting rid of irregular terms brought by the singular perturbation, we arrive at the tentative linear response formula for deterministic systems. The whole procedure will be similar to the {\em reverse} of finding the third root of a singular perturbation problem for polynomial $\epsilon x^3-x^2+1=0$ as $\epsilon\rightarrow0$ \cite{bender2013advanced}. 

Using the physical meaning of $S(t)$ and $\mathcal{T}(t)$, we write heuristically the response quantity $\mathcal{R}(s)$ for perturbation \eqref{FDR_perturbation_deter}:
\begin{equation}\label{R(s)_fluid_general_d}
\begin{aligned}
\mathcal{R}(s)=\sum_{l}\bigg(\frac{\delta\mu\nu|l|^4}{\sigma_l^2}-\frac{\delta\mu|l|^2}{2\langle\alpha_l^2\rangle_{\rho'}}\bigg)\alpha_l^2(s)
+
\left(\frac{\delta\mu\nu|l|^4}{\gamma_l^2}-\frac{\delta\mu|l|^2}{2\langle\beta_l^2\rangle_{\rho'}}\right)\beta_l^2(s)
-\frac{\delta\mu|l|^2}{\sigma_l}\alpha_l(s)B_l(s)
-\frac{\delta\mu|l|^2}{\gamma_l}\beta_l(s)W_l(s),
\end{aligned}
\end{equation}
one may find it the same as \eqref{R(s)_fluid_general}. According to the analysis of Matkowsky et al. \cite{matkowsky1977exit}, singular perturbation of the stationary Fokker-Planck equation has the following leading-order term (Formula (3.11) in \cite{matkowsky1977exit} with zero boundary condition $f(x)=0$):
\begin{align}\label{pdf}
    \rho_{\epsilon}\sim c_0-c_0e^{\zeta(x)/\epsilon^2},
\end{align}
where $\rho_{\epsilon}$ is the solution of stationary Fokker-Planck equation which corresponds to steady state distribution of the singularly perturbed system, while $c_0$ is the solution of the unperturbed hyperbolic equation which corresponds to steady state distribution of the deterministic  system. By eliminating from \eqref{pdf} the exponentially divergent leading term, we can obtain $c_0$. Assuming that similar result holds in the functional probability density setting, since the Radon-Nikodym derivative \eqref{Radon-Nikodym_derivative} implies 
\begin{align}\label{Radon-Nikodym_derivative1}
    P[x^{\delta}(t)]= \exp\left\{\frac{1}{2}\S(t)-\frac{1}{2}\mathcal{T}(t)\right\}P[x(t)]=\exp\left\{\int_0^t\mathcal{R}(s)ds\right\}P[x(t)]+O(\delta^2).
\end{align}
We identify that the leading order divergent terms brought by the singular perturbation are
\begin{align*}
\exp\left\{\sum_l\frac{\delta\mu\nu|l|^4}{\sigma_l^2}\int_0^t\alpha_l^2(s)ds\right\}, \qquad
\exp\left\{\sum_l\frac{\delta\mu\nu|l|^4}{\gamma_l^2}\int_0^t\beta_l^2(s)ds\right\}.
\end{align*}
Eliminating them from \eqref{R(s)_fluid_general_d} and then taking the limit $\sigma_l,\gamma_l\rightarrow 0$, we obtain linear response quantity \eqref{R(s)_fluid_NESS_deter} for the deterministic turbulence. Note that if constants $\eta=\theta=0$, we get the NESS1 case. If $\eta\neq0$ or $\theta\neq 0$, then we have the NESS3 case. Similarly, a tentative nonlinear response formula for deterministic turbulence would be $\langle O(t)\rangle_{\rho_{\delta}}=\langle\N(t)O(t)\rangle_{\rho}$, where $\N(t)$ is the change of the total entropy production given by:
\begin{align}\label{N(t)_deter}
\N(t)=\exp\left\{-\int_0^tds\sum_l\frac{\delta\mu|l|^2}{2\langle\alpha_l^2\rangle_{\rho'}}\alpha_l^2(s)+
\frac{\delta\mu|l|^2}{2\langle\beta_l^2\rangle_{\rho'}}\beta_l^2(s)
\right\}.
\end{align}

\bibliographystyle{plain}
\bibliography{MZ_Turbulence_dispersion}
\end{document}